\documentclass{article}

\usepackage{cancel}
\usepackage{bbm}

\usepackage{dsfont}
\usepackage{enumerate}
\usepackage{bm}
\usepackage[numbers]{natbib}  
\usepackage{amsmath}
\usepackage[utf8]{inputenc}
\usepackage[english]{babel}
\usepackage{amsthm}
\usepackage{amssymb}
\usepackage{algpseudocode,algorithm,algorithmicx}
\usepackage{graphicx}
\usepackage{mathtools}

\DeclareMathOperator*{\argmin}{arg\,min}
 
\newtheorem{theorem}{Theorem}

\newtheorem{lemma}{Lemma}
\newtheorem{definition}{Definition}
\newtheorem{proposition}{Proposition}

\algrenewcommand\algorithmicrequire{\textbf{Initialisation:}}
\algrenewcommand\algorithmicensure{\textbf{Output:}}

\DeclarePairedDelimiter\floor{\lfloor}{\rfloor}

\begin{document}

\title{Accelerating Parallel Tempering: Quantile Tempering Algorithm (QuanTA) 
}


\author{Nicholas G. Tawn and Gareth O.\ Roberts 
}



\maketitle

\begin{abstract}
  ~~Using MCMC to sample from a target distribution, $\pi(x)$ on a $d$-dimensional state space can be a difficult and computationally expensive problem. Particularly when the target exhibits multimodality, then the traditional methods can fail to explore the entire state space and this results in a bias sample output. Methods to overcome this issue include the parallel tempering algorithm which utilises an augmented state space approach to help the Markov chain traverse regions of low probability density and reach other modes. This method suffers from the curse of dimensionality which dramatically slows the transfer of mixing information from the auxiliary targets to the target of interest as $d \rightarrow \infty$. This paper introduces a novel prototype algorithm, QuanTA, that uses a Gaussian motivated transformation in an attempt to accelerate the mixing through the temperature schedule of a parallel tempering algorithm. This new algorithm is accompanied by a comprehensive theoretical analysis quantifying the improved efficiency and scalability of the approach; concluding that under weak regularity conditions the new approach gives accelerated mixing through the temperature schedule. Empirical evidence of the effectiveness of this new algorithm is illustrated on canonical examples.
	
	
\end{abstract}


\section{Introduction}

Consider the problem of stochastic simulation from a target distribution, $\pi(x)$ on a $d$-dimensional state space $\mathcal{X}$ where $\pi(\cdot)$ is known up to a scaling constant. The {\em gold standard} methodology for this problem
uses Markov chain Monte Carlo (MCMC). However these methods often perform poorly in the context of multimodality.

Most MCMC algorithms use localised proposal mechanisms, tuned towards local approximate optimality e.g.,\ \cite{roberts1997weak}, \cite{roberts2001optimal}. Indeed many MCMC algorithms incorporate local gradient information in the proposal mechanisms, typically attracting the chain back towards the centre of the mode. This can exacerbate the difficulties of moving between modes, \cite{Mangoubi2018}.

Popular methods used to overcome these issues include {\em simulated tempering}, \cite{marinari1992simulated} and the population-based version, {\em parallel tempering}, \cite{geyer1991markov}, \cite{Geyer1995}. These methods use state space augmentation to allow Markov chains to explore target distributions proportional to $\pi ^{\beta } (x)$ for $\beta$ typically in the range $(0,1]$. For simulated tempering this is done by
introducing an auxiliary {inverse temperature} variable, $\beta $, and running a $(d+1)$-dimensional Markov chain on $\mathcal{X}\times \Delta$,
where $\Delta $ consists of a discrete collection of possible inverse temperatures including $1$. For the more practically applicable parallel tempering approach, a Markov chain is run on a $\left(|\Delta|\times d\right)$-dimensional state space, $\mathcal{X}^{|\Delta|}$, where $|\Delta|$ denotes the cardinality of the set $\Delta$.

 Within this paper we will concentrate on parallel tempering as it obviates the need to approximate certain normalisation constants to work effectively. While parallel tempering has been highly successful, for example see \cite{Mohamed2012}, \cite{Xie2010}, \cite{Carter2013} etc, its efficiency declines as a function of $d$, at least linearly and often much worse \cite{atchade2011towards} and \cite{woodard2009sufficient}. This is caused by the need to set inter-inverse temperature spacings in $\Delta $ 
 extremely small to make swaps between temperatures feasible.

 This paper will introduce and analyse the QuanTA algorithm which facilitates inter-temperature swaps by proposing moves which attempt to adjust within-mode variation appropriately for the proposed new temperature. This leads to improved temperature mixing, which in turn leads to vastly improved inter-modal mixing. Its typical improvement is demonstrated in Figure~\ref{Fig:Onedimreparmex} with a 5-mode target distribution.
\begin{figure}[h]
\begin{center}
\includegraphics[keepaspectratio,width=9cm]{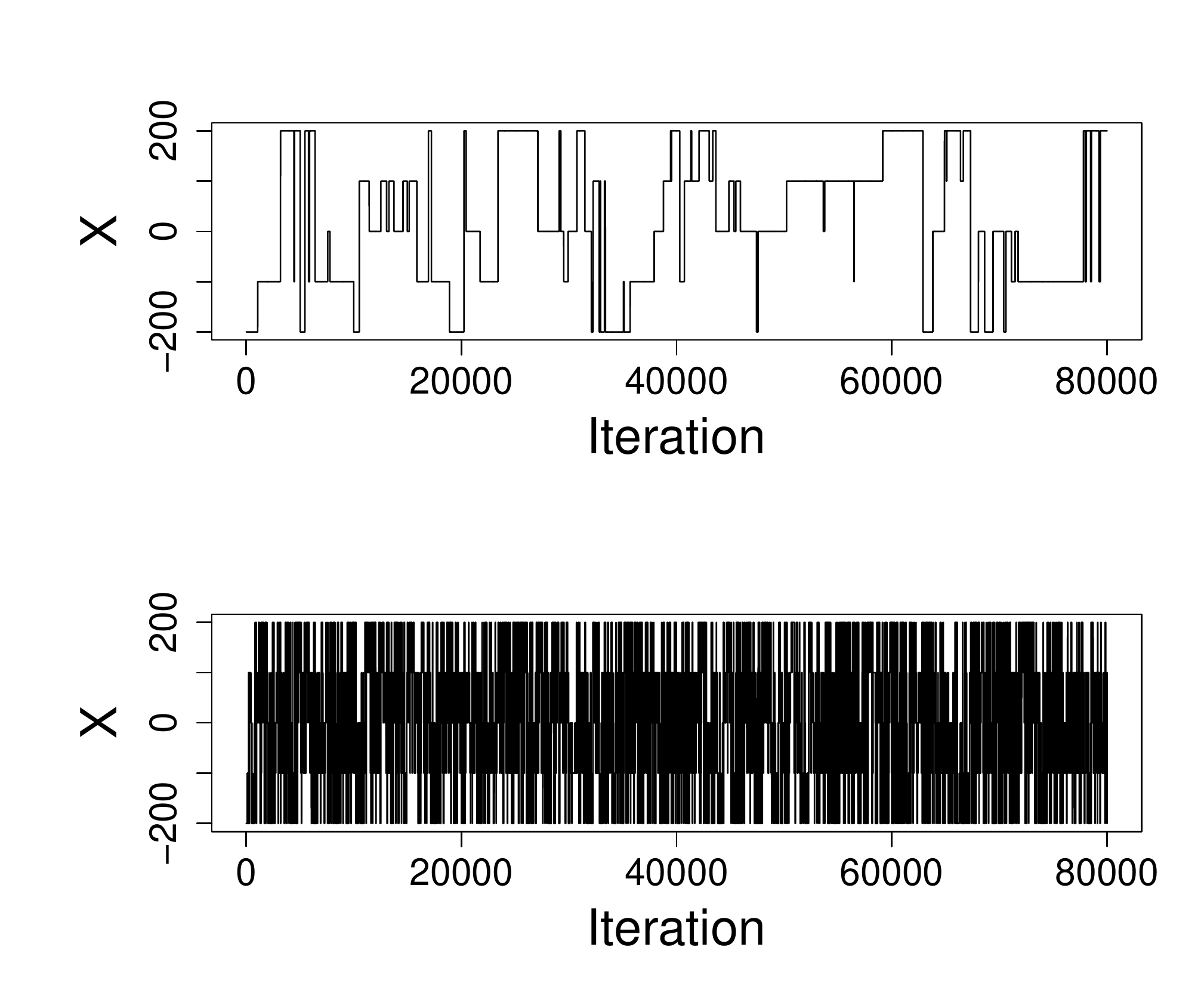}
\caption{Trace plots of the target state chains for representative runs of the Parallel Tempering (top) and QuanTA schemes (bottom). 
}
\label{Fig:Onedimreparmex}
\end{center}
\end{figure}
The construction of QuanTA resonates with the non-centering MCMC methodology described, for example in \cite{bernardo2003non}, \cite{brooks2003efficient}, \cite{hastie2005towards} and \cite{hierachicalparamRoberts}. 

Supporting theory is developed to guide setup and analyse the utility of the novel QuanTA scheme. There are two key theoretical results. The first, Theorem~\ref{Thr:optscalQuanTA}, establishes that there is an optimal temperature schedule setup for QuanTA; concluding that in general the dimensionality scaling of the distance between consecutive inverse temperature spacings should be $\mathcal{O}(d^{-1/2})$.  Further to this it suggests that optimising the expected squared jumping distance between any two consecutive temperature levels induces a temperature swap move acceptance rate of 0.234; giving a useful metric for a practitioner to optimally tune QuanTA. The second key theoretical contribution, Theorem~\ref{cor:higherorder}, of this paper shows that, under mild regularity conditions, the optimal temperature spacings of QuanTA are more ambitiously spaced than for the standard parallel tempering algorithm for cold (i.e.\ large) values of the inverse temperatures. The significance of this result is that QuanTA can give accelerated mixing through the cooler parts of the temperature schedule by allowing more ambitious temperature spacings.

This paper is structured into 6 core sections. Sections~\ref{subsec:parallel} reviews the parallel tempering algorithm and some of the relevant existing literature. Section~\ref{sec:quantpressec} motivates the main idea behind the novel QuanTA scheme, which is then presented in Section~\ref{subset:newalg}. QuanTA utilises a population MCMC approach that requires a clustering scheme; discussion for this is found in Section~\ref{subsec:locmodgauss}. Section~\ref{sec:optscale} contains the core  theoretical contributions mentioned above.  Simulation studies are detailed in Section~\ref{subsec:Examples} along with a discussion of the computational complexity of QuanTA.

\section{The Parallel Tempering (PT) Algorithm}
\label{subsec:parallel}

There is an array of methodology available to overcome the issues of multimodality in MCMC,  the majority of which use state space augmentation e.g.\ \cite{Wang1990a}, \cite{geyer1991markov}, \cite{marinari1992simulated}, \cite{Neal1996}, \cite{kou2006discussion}, \cite{2017arXiv170805239N}. Auxiliary distributions that allow a Markov chain to explore the entirety of the state space are targeted and their mixing information is then passed on to aid  inter-modal mixing in the desired target. A  convenient approach for the augmentation methods is to use power-tempered target distributions i.e.,\ the target distribution at inverse temperature level, $\beta$, for $\beta \in (0,1]$ is defined as \[\pi_\beta(x)\propto \left[\pi(x)\right]^\beta\] 

Such targets are the most common choice of auxiliary target when augmenting the state space for use in the popular simulated tempering (ST) and parallel tempering (PT) algorithms introduced in \cite{marinari1992simulated} and \cite{geyer1991markov}. For each algorithm one needs to choose a sequence of $n+1$ ``inverse temperatures'', $\Delta=\{\beta_0,\ldots,\beta_n\}$, where $0 \leq \beta_n<\beta_{n-1}<\ldots <\beta_1<\beta_0=1$ with the specification that a Markov chain sampling from the target distribution $\pi_{\beta_n}(x)$ can mix well across the entire state space.

 The PT algorithm runs a Markov chain on the augmented state space, $\mathcal{X}^{(n+1)}$, targeting an invariant distribution given by
\begin{eqnarray}
\pi_n(x_0,x_1,\ldots,x_n) \propto \pi_{\beta_0}(x_0)\pi_{\beta_1}(x_1)\ldots\pi_{\beta_n}(x_n).\label{eq:PTtarg111}
\end{eqnarray}
From an initialisation point for the chain the PT algorithm  alternates between two types of Markovian move. \textit{Within temperature} Markov chain moves that use standard localised MCMC schemes to update each of the $x_i$ whilst preserving marginal invariance. \textit{Temperature swap} moves that propose to swap the chain locations between a pair of adjacent temperature components. It is these swap moves that will allow mixing information from the hot, rapidly-mixing temperature level to be passed to aid mixing at the cold target state. 

To perform the swap move a pair of temperatures is chosen uniformly from the set of all adjacent pairs, call this pair $x_i$ and $x_{i+1}$ at inverse temperatures $\beta_{x_i}$ and $\beta_{x_{i+1}}$ respectively. The proposal is then
\begin{equation}
(x_0,\ldots,x_i,x_{i+1},\ldots,x_n) \rightarrow (x_0,\ldots,x_{i+1},x_i,\ldots,x_n) \label{eq:swapmove}
\end{equation}
To preserve detailed balance and therefore invariance to $\pi_n(\cdot)$, the swap move is accepted with probability
\begin{equation}
A=\mbox{min}\Bigg( 1,\frac{\pi_{\beta_{x_{i+1}}}(x_i)\pi_{\beta_{x_i}}(x_{i+1})}{\pi_{\beta_{x_i}}(x_i)\pi_{\beta_{x_{i+1}}}(x_{i+1})} \Bigg).
\label{eq:parstd1}
\end{equation}

It is the combination of the suitably specified \textit{within temperature} moves and \textit{temperature swap} moves that ensures ergodicity of the Markov chain to the target distribution, $\pi_n(\cdot)$. Note that the \textit{within temperature} moves certainly influence the performance of the algorithm, \cite{Ge2017}; however the focus of the work in this article will be on designing a novel approach for the \textit{temperature swap} move.

The novel work presented in this paper focuses on the setting where the $d$-dimensional state space is given by $\mathbb{R}^d$ and the target, $\pi(\cdot)$, is the associated probability density function. Thus, herein take  $\mathcal{X} = \mathbb{R}^d$ but note that natural generalisations to other state spaces and settings are possible.

\section{Modal Rescaling Transformation}
\label{sec:quantpressec}

\subsection{A Motivating Transformation Move}
\label{subsec:heurist}

Consider a PT algorithm that has two components $x_1$ and $x_2$ running at the neighbouring inverse temperature level $\beta$ and $\beta^{'}$. Suppose that a temperature swap move is proposed between the two chains at the two temperature levels. Due to the dependence between the location in the state space and the temperature level, $\beta$ and $\beta^{'}$ need to be close to each other to avoid the move having negligible acceptance probability. Intuitively, the problem is that the proposal from the hotter chain is likely to be an ``unrepresentative'' location at the colder temperature and vice versa.

So there is clearly a significant dependence between the temperature value and the location of the chain in the state space; thus explaining why temperature swap moves between arbitrarily-largely spaced temperatures are generally rejected. This issue is typically exacerbated when the dimensionality grows.

Consider for motivational purposes, a simple one-dimensional setting where the state space is given by $\mathbb{R}$ and the target density is given by $\pi(\cdot)$. For notational convenience letting $j=i+1$, suppose that a temperature swap move has been proposed between adjacent levels $\beta_i$ and  $\beta_j$ with marginal component values $x_i$ and $x_j$ respectively.

Suppose an oracle has provided a function,  $g_{ij}: \mathbb{R}\rightarrow \mathbb{R}$, that is  bijective, with $g_{ji}(g_{ij}(x))=x$, and  differentiable and preserves the CDF between the two temperature levels such that
\begin{equation}
F_{\beta_j}(g_{ij}(x))=F_{\beta_i}(x).
\label{eq:quanpres}
\end{equation}
So suppose that rather than the standard temperature swap move proposal in \eqref{eq:swapmove}, the following is instead proposed:
\begin{equation}
(x_0,\ldots,x_i,x_j,\ldots,x_n) \rightarrow (x_0,\ldots,g_{ji}(x_j),g_{ij}(x_i),\ldots,x_n) \label{eq:gswapmove}
\end{equation}

To preserve detailed balance this is accepted with an acceptance ratio similar to reversible-jump MCMC, \cite{green1995reversible}, to account for the deterministic transformation:
\begin{eqnarray}
\mbox{min}\left( 1,\frac{\pi^{\beta_j}(g_{ij}(x_i))\pi^{\beta_i}(g_{ji}(x_j))}{\pi^{\beta_i}(x_i)\pi^{\beta_j}(x_j)}\left|  \frac{\partial g_{ij}(x_i)}{\partial x}  \right|\left|  \frac{\partial g_{ji}(x_j)}{\partial x}  \right|\right). \label{eq:quantpresaccep}
\end{eqnarray}
A simple calculation using \eqref{eq:quanpres} shows that this equals one and hence such a swap would always be accepted. Essentially, the transformation $g_{ij}(\cdot)$ has made the acceptance probability of a temperature swap move independent of the locations of $x_i$ and $x_j$ in the state space.

In practice, a CDF-preserving function $g_{ij}(\cdot)$ will not generally be available. Consider a simplified setting when the target is now a $d$-dimensional Gaussian, i.e.\ $\pi\sim N(\mu,\Sigma)$, and so the tempered target at inverse temperature $\beta$ is given by $\pi^\beta \sim N(\mu,\Sigma /\beta)$. Defining a $d$-dimensional transformation by
\begin{equation}
g_{ij}(x,\mu)=\left(\frac{\beta_i}{\beta_j}\right)^{1/2}\left(x-\mu  \right)+\mu, \label{eq:gausspresnew}
\end{equation}
a simple calculation shows that in this setting such a transformation, which only requires knowledge of the mode location, permits swap moves to always be accepted independently of the dimensionality and magnitude of the inverse temperature spacings.

In a broad class of applications it is not unreasonable to make a Gaussian approximation to posterior modes, \cite{Rue2009}. Indeed this is the motivation for the similar transformation derived in  \cite{hastie2005towards} for use in a reversible-jump MCMC framework.

\subsection{Transformation move in a PT Framework}
\label{subsec:transmov}

In a multimodal setting a single Gaussian approximation to the posterior will be poor. However, it is often reasonable that the local modes may be individually approximated as Gaussian. This paper explores the use of the transformation in \eqref{eq:gausspresnew} applied to the local mode with the aim being to accelerate the mixing through the temperature schedule of a PT algorithm.

Now that the transformations are localised to modes one needs careful specification of the transformation function. Suppose that there is a collection of $K$ mode points, $\mu_1,\ldots,\mu_K$ and a metric, $m(x,y)$ for $x,y \in \mathbb{R}^d$, that will be used to associate locations in the state space with a mode. To this end define the mode allocating function
\begin{equation}
	Z(x) = \argmin_{h \in \{1,\ldots,K\}} \left[ m(x,\mu_h) \right] \nonumber
\end{equation}
and with $g_{ij}(\cdot)$ from \eqref{eq:gausspresnew} define the sets
\begin{equation}
	A_{ij} = \left\{ x \in \mathbb{R}^d :  Z(g_{ij}(x,\mu_{Z(x)}))=Z(x)  \right\} \label{Aij}
\end{equation}
define the transformation, 
\begin{equation}
g(x,\beta_i,\beta_j) =  g_{ij}(x,\mu_{Z(x)}) 
. \label{eq:fullydefinedtrans}
\end{equation}

The aim is to use this transformation in a PT framework. So suppose that a temperature swap move proposal is made between two marginal components $x_i$ and $x_j$ at respective inverse temperatures $\beta_i$ and $\beta_j$ with $\beta_i>\beta_j$. The idea is that this swap move now utilises  \eqref{eq:fullydefinedtrans} so that the proposed move takes the form
\begin{equation}
(x_0,\ldots,x_j,x_i,\ldots,x_n) \rightarrow (x_0,\ldots,g(x_i, \beta_i,\beta_j),g(x_j,, \beta_j,\beta_i),\ldots,x_n) 
\label{eq:gggswapmove}
\end{equation}
which to satisfy detailed balance is accepted with probability
\begin{eqnarray}
~~~~~\mbox{min}\left( 1,\frac{\pi(g(x_i,\beta_i,\beta_j))^{\beta_{j}}\pi(g(x_j,\beta_j,\beta_i))^{\beta_i}}{\pi(x_i)^{\beta_i}\pi(x_{j})^{\beta_{j}}} \mathbbm{1}_{\left\{x_i \in A_{ij}\right\}} \mathbbm{1}_{\left\{x_j \in A_{ji}\right\}}\right). \label{eq:quantaccrat}
\end{eqnarray}

\begin{proposition}
Consider a Markov chain that is in stationarity with a target distribution given by \eqref{eq:PTtarg111} on a state space $\mathcal{X}=\mathbb{R}^d$. Let $\mu_1,\ldots,\mu_K \in \mathbb{R}^d$. If a temperature swap move of the form \eqref{eq:gggswapmove} is proposed where the transformation is given by \eqref{eq:fullydefinedtrans} and is accepted with probability given in \eqref{eq:quantaccrat} then the chain is invariant with respect to \eqref{eq:PTtarg111}.
\end{proposition}

Of course to make this transformation one needs the collection of $K$ centring points.  Essentially these are attained through use of an appropriate clustering procedure; suggestions using population MCMC methods are given in Section~\ref{subsec:locmodgauss}.

\section{Quantile Tempering Algorithm (QuanTA)}
\label{subset:newalg}

Motivated by the calculations in Section~\ref{sec:quantpressec}, QuanTA is introduced to exploit the use of the transformation move established in \eqref{eq:fullydefinedtrans} and \eqref{eq:quantaccrat}.

QuanTA runs the equivalent of $N$ parallel tempering algorithms procedures in parallel with each single procedure using the same tempering schedule. With a temperature schedule given by $\Delta=\{ \beta_0,\ldots,\beta_n\}$, the QuanTA approach can be seen as running a single Markov Chain on the augmented state space, $(\mathbb{R}^d)^{n*N}$. Denoting $\mathbf{x}= (x_{(1,0)},\ldots,x_{(1,n)},x_{(2,0)},\ldots,x_{(N,n)})$, the invariant target distribution for the Markov chain induced by QuanTA is
\begin{eqnarray}
\pi_Q(\mathbf{x}) \propto \prod_{i=1}^N\pi_{\beta_0}\left(x_{(i,0)}\right)\pi_{\beta_1}\left(x_{(i,1)}\right)\ldots\pi_{\beta_n}\left(x_{(i,n)}\right).\nonumber
\end{eqnarray} 

\textbf{Initialisation:} to initialise the QuanTA algorithm, one is required to choose: initial starting values for the Markov chain components; a suitable temperature schedule (see Theorem~\ref{Thr:optscalQuanTA} in Section~\ref{sec:scalingstheorem} for suggested optimality criteria for the temperature schedule); the size of $N$ and suitable parameters for the chosen clustering method that will be used.

\textbf{Running the chain:} from the start point of the chain, QuanTA  alternates between two types of Markov chain moves. 

\textbf{\textit{Within temperature}} Markov chain moves that use standard localised MCMC schemes for marginal updates of each of the $x_{(i,j)}$. Essentially, this is just Metropolis-within-Gibbs MCMC and in this setting, with hugely exploitable marginal independence, this process is highly parallelisable. Denote the $\pi_Q$-invariant Markov transition kernel that performs temperature marginal updates on all components from a current point $\mathbf{x}$ as $P_1(\mathbf{x},d\mathbf{y})$.

\textbf{\textit{Temperature swap}} moves that propose to swap the chain locations between a pair of adjacent temperature components. This is where QuanTA differs from the standard PT procedure and uses the new transformation aided temperature swap move detailed in Section~\ref{subsec:transmov} in particular in \eqref{eq:fullydefinedtrans}. This follows a two phase population-MCMC update procedure. 
\begin{itemize}
	\item \textbf{Phase 1:} Group marginal components into two collections, 
	\begin{eqnarray}
	C_1 &=& \{ x_{(i,j)}: i = 1,\ldots, \floor*{N/2} \text{~and~} j=0,\ldots,n\} \nonumber \\ C_2 &=& \{ x_{(i,j)}: i = (\floor*{N/2}+1),\ldots,N  \text{~and~} j=0,\ldots,n\}.\nonumber
	\end{eqnarray}
	
	An appropriate clustering scheme (see Section~\ref{subsec:locmodgauss}) is performed on $C_1$ providing a set of $K$ centres $\{c_1,\ldots,c_K\}$. To enhance the effectiveness of the transformation it is suggested that these cluster centre points are used as initialisation locations for a  suitable local optimisation procedure to find $K$ mode points $M_1=\{\mu_1,\ldots,\mu_K\}$ of $\pi (\cdot)$ (see Theorem~\ref{cor:higherorder} in Section~\ref{sec:scalingstheorem}). 
	
	 For each $i \in \{(\floor*{N/2}+1),\ldots,N \}$, sample $l \sim Unif\{0,1,\ldots,n-1\}$ and select the corresponding pair of adjacent temperature marginals $(x_{(i,l)},x_{(i,l+1)})$ for a temperature swap move proposal utilising the transformation from \eqref{eq:fullydefinedtrans} (which is centred on the associated point from $M_1$). This move is accepted with probability \eqref{eq:quantaccrat}.

	\item \textbf{Phase 2:} Repeat phase 1 but with the roles of $C_1$ and $C_2$ reversed.
\end{itemize}
Denote the $\pi_Q$-invariant Markov transition kernel that implements this  temperature swap update procedure for all components using the above two-phase process by $P_2(\mathbf{x},d\mathbf{y})$.

From the initialisation point $\mathbf{x}$ then the Markov chain output is created by application of the following kernel compilation:
\begin{equation}
	(P_2\circ P_1^k)^T \nonumber
\end{equation}
where $k$ is the user-chosen number of within temperature Markov chain updates between each swap move proposal and $T$ is the user-chosen number of iterations of the algorithm before stopping.

\begin{proposition}
Provided suitable within-temperature MCMC moves are implemented the Markov chain constructed by QuanTA is $\pi_Q(\cdot)$ invariant.
\end{proposition}
\begin{proof}
Proof follows immediately since this is an instance of the Metropolis-within-Gibbs construction.
\end{proof}

\section{Estimating Local Mode Locations}
\label{subsec:locmodgauss}

The QuanTA algorithm, presented in Section~\ref{subset:newalg},  requires online estimation of the local mode points as centring location for the transformation. This section outlines a practical scheme that is used in the canonical simulation studies.

With a typically unknown number of modes and a population of chains, a principled approach  would be to fit a Dirichlet Process mixture model, e.g.\ \cite{neal2000markov} and \cite{kim2006variable}. A comprehensive Gibbs sampling approach for this can be computationally expensive, but there are alternative cheaper but approximate methods that are left for exploration in further work, \cite{raykov2016simple}.

For the examples with well-separated modes that were studied here it sufficed to use a cheap and fast clustering scheme, \cite{friedman2001elements}. To this end a K means approach was used, \cite{hartigan1979algorithm}. The clustering procedure provides a collection of cluster centres that  can be directly used as centring points for the transformation or as very  useful initialisation points for a local optimisation method. Indeed, Theorem~\ref{cor:higherorder} of Section~\ref{sec:optscale} shows that QuanTA can achieve accelerated mixing through the temperature levels when the centring point is chosen as the mode point, particularly at colder temperatures when the Gaussian approximation to the mode becomes increasingly accurate, e.g.\ \cite{Barndorff-Nielsen1989} and \cite{Olver1968}.

\subsection{A Weighted K Means Clustering}
\label{sec:Kmeans}

Typically the K means algorithm assigns all points equal leverage in determining cluster centres. A weighted K means approach incorporates weights that can alter the leverages of points. In the tempering setting chains at the colder states, where the modes are less disperse, should have more leverage in determining the centres.

Weighted K means is an almost identical procedure to the K means algorithm of \cite{hartigan1979algorithm} but now incorporates the weights to give points leverage. For the setting of interest each chain location will be allocated a weight, determined by their inverse temperature value. For a collection of $n$ chain locations $x_1,\ldots, x_n$ at inverse temperature levels $\beta_{x_1},\ldots,\beta_{x_n}$.  The weighted K means algorithm attempts to iteratively establish a particle allocation $S$ such that, each point $x_j$ to a cluster $S_k\in S= \{S_1,\ldots,S_K  \}$ with:
\begin{equation}
\mbox{argmin}_{S} \left\{ \sum_{i=1}^M \sum_{j=1}^n \mathbbm{1}_{\{x_j\in S_i\}} \beta_{x_j} ||x_j-\mu_{i}||^2 \right\}. \label{eq:objective11}
\end{equation}
The weighted K means algorithm begins with an initial set of K centres $\{\mu_1,\ldots,\mu_K\}$. It then proceeds by alternating between two updating steps until point allocations do not change (signalling a minimum of  or a pre-specified number of iterations is reached. A \textit{point allocation step}, where each point, $x$, is assigned to the set $S_j$ where $j=\argmin_j ||x-\mu_{j}||^2$. An \textit{centre point update step} where for the new allocation the centring points are each updated to be the weighted mean of their respective component steps, i.e.\ 
\begin{equation}
	\mu_i= \frac{\sum_{j\in S_i} x_j \beta_{x_j}}{\sum_{j\in S_i}  \beta_{x_j}}. \nonumber
\end{equation}

The Weighted K means procedure can be implemented using the R package``FactoClass'',  by \cite{factoclass2007combinacion} which uses a modified version of the K means algorithm of \cite{hartigan1979algorithm}.

\section{Theoretical Underpinnings of QuanTA in High Dimensions}
\label{sec:optscale}

In both QuanTA and the PT algorithms, the acceptance of temperature swap proposals allow the transfer of hot-state mixing information to be passed through to the cold state. The ambitiousness of the spacings between the consecutive inverse temperatures dictate the performance of the algorithm. Similarly to the problem of tuning the RWM algorithm, \cite{roberts1997weak}, one seeks the optimal balance between over and under-ambitious proposals. This issue becomes increasingly problematic with an increase in dimensionality, hence careful scaling of the consecutive temperatures spacings is needed to prevent degeneracy of acceptance rates.

 The work in \cite{atchade2011towards} sought an optimal scaling result for temperature spacings in a PT algorithm. This section takes a similar approach to derive an equivalent result for QuanTA. It will be shown in Theorem~\ref{Thr:optscalQuanTA} that consecutive spacings inducing swap rates of approximately 0.234 are optimal; thus giving guidance for practitioners to tune towards an optimal setup. Complementary to this, Theorem~\ref{cor:higherorder} justifies the use of QuanTA outside the Gaussian setting; showing that under mild conditions the transformation move allows for larger spacings in the temperature schedule than the PT algorithm does.

\subsection{Optimal Scaling of QuanTA- The Setup and  Assumptions}
\label{sec:setupopt}

As the dimensionality, $d$,  of the target distribution tends to infinity,  the problem of selecting temperature spacings for QuanTA is investigated.
Suppose a swap move between two consecutive temperature levels, $\beta$ and $\beta'=\beta+\epsilon$ for some $\epsilon>0$ is proposed. As in \cite{atchade2011towards}, the measure the efficiency of the inverse temperature spacing will be the expected squared jumping distance, $ESJD_\beta$, defined as
\begin{equation}
	ESJD_\beta = \mathbb{E}_{\pi_n} \left[ (\gamma-\beta)^2 \right] \label{equ:ESJDbeta1}
\end{equation}
where $\gamma=\beta+\epsilon$  if the proposed swap is accepted and $\gamma=\beta$ otherwise. Note the assumption that the Markov chain has reached invariance and so  the expectation is taken with respect to the  invariant distribution, $\pi_n(\cdot)$. 

The $ESJD_\beta$ is a natural quantity to consider, \cite{Sherlock2006}, since maximising this would appear to ensure that one is being sufficiently ambitious with spacings but not inducing degenerate acceptance rates. However, it is worth noting that it is only truly justified  when there is an associated diffusion limit for the chain, \cite{roberts2014minimising}. 

The aim is to establish the limiting behaviour of the $ESJD_\beta$ as $d\rightarrow \infty$ and then optimise this limiting form. To this end, for tractability, the form of the  $d$-dimensional target is restricted to distributions of the form:
\begin{equation}
\pi(x) \propto f_d(x) = \prod_{i=1}^d f(x_i).
\label{eq:restrict}
\end{equation}
and to achieve a non-degenerate acceptance rate as $d \rightarrow \infty$ the spacings are necessarily scaled as $\mathcal{O}(d^{-1/2})$, i.e.\
\begin{equation}
\epsilon = \frac{\ell}{d^{1/2}}.
\label{eq:epsform1}
\end{equation}
where $\ell$ a positive constant that one tunes to attain an optimal $ESJD_\beta$.

Furthermore, assume that the univariate marginal components, $f(\cdot)$, are $C^4$ and unimodal with a maximum at $\mu$.  Furthermore, the marginal components $f(\cdot)$ are assumed to be of the form
\begin{equation}
	f(x)=e^{-H(x)} ~~~~ \forall x \in \mathbb{R}\label{eq:formoff}
\end{equation}
where the $H(x):=-\log(f(x))$ is regularly varying, \cite{Bingham1989} i.e.\ there exists an $\alpha >0$ such that for $x>0$
\begin{equation}
	\frac{H(tx)}{H(t)}\rightarrow x^{\alpha} ~~\mbox{as}~~ |t|\rightarrow \infty. \label{eq:reguvar}
\end{equation}
This is a sufficient condition for Theorem~\ref{Thr:optscalQuanTA} and  ensures the moments and integrals required for the proof are all well defined. Further assume that the fourth derivatives of $(\log f)(\cdot)$ are bounded, i.e. $\exists M>0$ such that
\begin{equation}
	| (\log f)''''(z)| <M ~~~ \forall z \in \mathbb{R}. \label{eq:quantafourthass}
\end{equation}
This condition is sufficient for proving Theorem~\ref{Thr:optscalQuanTA} but not necessary. The proof still works if the condition is weakened so that for some $k\ge 4$ then the $k^{th}$ derivative of the logged density is bounded.

Finally, for notational convenience, the following are defined, with the subscript $\beta$ indicating that the expectation is with respect to $f^\beta (\cdot)$:
\begin{eqnarray}
V(\beta)&=&\mathrm{Cov}_{\beta}((\log f)(x),(x-\mu)(\log f)'(x))=\frac{1}{\beta^2} \nonumber \\
I(\beta)&=&\mathrm{Var}_{\beta}\left[(\log f)(x)\right]\nonumber \\
R(\beta)&=&\mathbb{E}_{\beta}\left[(x-\mu)^2(\log f)''(x)-(x-\mu)(\log f)'(x)\right]. \nonumber
\end{eqnarray}

Note it is assumed that the univariate marginal components, $f(\cdot)$, are unimodal. This is  a significant and strong assumption. The problem is that the allocation to a mode point essentially splits the state space into regions, and the mass in each region can be dramatically inconsistent between consecutive temperature levels, \cite{woodard2009sufficient}. This would result in a degenerate limit to the $ESJD_\beta$ in this setting.  This doesn't mean that the results presented are invalid for multimodal situations since when the modes are all well separated and identical in form, then without loss of generality it can be assumed that both particles are in the same mode.

Due to the uni-modality there is a simplified form of the acceptance probability that no-longer requires the indicator functions. 
Denote the acceptance probability of the QuanTA-style swap move by $\alpha_{\beta}(x,y)$ so
\begin{equation}
\alpha_{\beta}(x,y)=\mbox{min}\left(1, \frac{f_d^{\beta'}(g(x,\beta,\beta^{'}))f_d^{\beta}(g(y,\beta^{'},\beta))}{f_d^{\beta'}(y)f_d^{\beta}(x)} \right), \label{eq:accscaling}
\end{equation}
then a simple calculation shows that the $ESJD_\beta$ from \eqref{equ:ESJDbeta1},  becomes
\begin{eqnarray}
ESJD_{\beta} 
= \epsilon^2 \mathbb{E}_{\pi_n} \left[ \alpha_{\beta}(x,y) \right] \label{eq:ESJDP12form} 
\end{eqnarray}
which will be maximised with respect to $\ell$ in the limit as $d\rightarrow \infty$.

\subsection{ Scaling Results and Interpretation}
\label{sec:scalingstheorem}

Under the setting of Section~\ref{sec:setupopt} and with $\Phi(\cdot)$ denoting the CDF of a standard Gaussian, the following optimal scaling result is derived:
\begin{theorem}[Optimal Scaling for the QuanTA Algorithm]
Consider QuanTA targeting a distribution, $\pi(\cdot)$, satisfying \eqref{eq:restrict}. Assume that the marginal components, $f(\cdot)$, are regularly varying, satisfying \eqref{eq:formoff} and \eqref{eq:reguvar}, unimodal, and  $\log f(\cdot)$ satisfies \eqref{eq:quantafourthass}. Assuming $\epsilon = \ell / d^{1/2}$ for some  $\ell\in \mathbb{R}_{+}$ then in the limit as $d\rightarrow\infty$, the $ESJD_{\beta}$, given in \eqref{eq:ESJDP12form}  is maximised when $\ell$ is chosen to maximise
\begin{equation}
2\ell^2\Phi\left(-\frac{\ell\left[\frac{1}{2} V(\beta) -I(\beta)+\frac{1}{4\beta}R(\beta)\right]^{1/2}}{\sqrt{2}}\right), \label{eq:P1theoremstateesjd}
\end{equation}

Furthermore, for the optimal $\ell$ the corresponding swap move acceptance rate induced between two consecutive temperatures is given by 0.234 (3.s.f).
\label{Thr:optscalQuanTA}
\end{theorem}

\begin{proof}
The details of the proof of Theorem~\ref{Thr:optscalQuanTA} are deferred to the Appendix, Section~\ref{sec:proofQuanTA}. The strategy  comprises 3 key stages which are: establishing a Taylor  series expansion of the logged swap move acceptance ratio (i.e.\ the log of \eqref{eq:accscaling});  establishing limiting Gaussianity of this logged acceptance ratio; and finally, achieving a tractable form of the limiting $ESJD_{\beta}$ which is then optimised with respect to $\ell$ giving rise to an associated optimal acceptance rate. 
\end{proof}

\textit{Remark 1:}  In the special case that the marginal targets are Gaussian, i.e.\ $f(x) = \phi(x; \mu, \sigma^2)$ then the transformation swap move should permit arbitrarily ambitious spacings. This is verified by observing that in this case 	
\begin{equation}
\left[\frac{1}{2} V(\beta) -I(\beta)+\frac{1}{4\beta}R(\beta)\right]=0 \nonumber
\end{equation}
and so with respect to $\ell$ \eqref{eq:P1theoremstateesjd} becomes proportional to $\ell^2$ which has no finite maximal value; thus demonstrating  consistency with what is know in the Gaussian case.

\textit{Remark 2:} The optimality criterion given in \eqref{eq:P1theoremstateesjd} is very similar to that derived in \cite{atchade2011towards} and \cite{roberts2014minimising}. Indeed, both QuanTA and the PT algorithm require the same dimensionality spacing scaling and both are optimised when a 0.234 acceptance rate is induced. However, there will be a difference in the behaviour of the optimal $\hat{\ell}$ which is where QuanTA can be shown to give accelerated mixing versus the PT approach, see Theorem~\ref{cor:higherorder} below.

\textit{Remark 3:} Theorem~\ref{Thr:optscalQuanTA} gives an explicit formula for derivation of the optimal $\hat{\ell}$ between consecutive temperatures but this is usually intractable in a real problem. However, for a practitioner, the associated 0.234 optimal swap acceptance rate gives useful setup guidelines. In fact, the theorem suggest a strategy for optimal setup starting with a chain at the hottest level and tuning the spacing to successively colder temperature levels based on the swap acceptance rate to attain consecutive swap rates close to 0.234. Indeed, using a stochastic approximation algorithm, see \cite{robbins1951stochastic}, then \cite{miasojedow2013adaptive} took an adaptive MCMC approach, \cite{roberts2009examples}, to do this for the PT algorithm but their framework also extends naturally to QuanTA.

\subsubsection{Higher Order Scalings at Cold Temperatures}
\label{sec:supercoldscalings}

For any univariate Gaussian distribution at inverse temperature level $\beta$, $I(\beta)=1/(2\beta^2)$. It is shown in \cite{atchade2011towards} that the optimal choice for the scaling parameter takes the form
\begin{equation}
	\hat{\ell}\propto I(\beta)^{-1/2} \propto \beta \label{eq:thebasecase}
\end{equation}
resulting in a geometrically spaced temperature schedule.

Assuming appropriate smoothness for the marginal components , $f(\cdot)$, then for a sufficiently cold temperature the local mode can be well approximated by a Gaussian.  So for sufficiently cold temperatures one expects $I(\beta) \approx 1/(2\beta^2)$; thus spacings become (approximately) $\mathcal{O}(\beta)$ (note that a rigorous derivation that $I(\beta) \approx 1/(2\beta^2)$ is contained in the proof of Theorem~\ref{cor:higherorder}). Defining the ``order of the spacing with respect to the inverse temperature, $\beta$'' as the value of $\zeta$ such that the optimal spacing is $\mathcal{O}(\beta^\zeta)$ then the standard PT algorithm is order 1 for sufficiently cold temperatures.

In the Gaussian setting,  QuanTA  exhibits  ``infinitely'' high order behaviour since there is no restriction on the size of the temperature spacings with regards the value of $\beta$.  It is hoped that some of this higher order behaviour is inherited in a more general target distribution setting when the target is cooled and increasingly approaches Gaussianity. Indeed, under the setting of  Theorem~\ref{Thr:optscalQuanTA} but with a single additional condition it is shown that QuanTA does exhibit higher order behaviour than the PT algorithm at cold temperatures.

With $f(\cdot)$ as in Theorem~\ref{Thr:optscalQuanTA} ( but now without loss of generality the mode point is at $\mu=0$) define the normalised density $g_\beta(\cdot)$ as
\begin{equation}
	g_\beta (y) \propto f^\beta \left(\mu+ \frac{y}{\sqrt{-\beta(\log f)''(\mu)}}\right)=f^\beta \left( \frac{y}{\sqrt{-\beta(\log f)''(0)}}\right).  \label{eq:gformtherscal}
\end{equation}

The additional assumption required to prove the higher order behaviour of QuanTA is that there exists $\gamma>0$ such that as $\beta\rightarrow \infty$
\begin{equation}
	| \mathrm{Var}_{g_\beta}\left( Y^2  \right) -2 | = \mathcal{O}\left( \frac{1}{\beta^\gamma} \right). \label{eq:assumptionscale}
\end{equation}
This assumption essentially guarantees the convergence to Gaussianity about the mode as $\beta\rightarrow\infty$. This assumption appears to be reasonable with studies of both a Gamma and a student-t distributions demonstrating a value of $\gamma=1$; details can be found in \cite{NTawnThesis}.

\begin{theorem}[Cold Temperature Scalings]\label{cor:higherorder}
For marginal targets, $f(\cdot)$, satisfying the conditions of Theorem~\ref{Thr:optscalQuanTA} and \eqref{eq:assumptionscale}, then  for $\beta$ sufficiently large 
\begin{equation}
	\left[\frac{1}{2} V(\beta) -I(\beta)+\frac{1}{4\beta}R(\beta)\right] = \mathcal{O}\left(\frac{1}{\beta^k}\right), \nonumber
\end{equation}
where 
\begin{itemize}
	\item $k=\min\left\{2+\gamma, 3 \right\}>2$ if $f$ is symmetric about the mode point 0
	\item $k=\min\left\{2+\gamma, \frac{5}{2} \right\}>2$ otherwise.
\end{itemize}
This induces an optimising value $\hat{\ell}$ such that
\begin{equation}
	\hat{\ell}= \mathcal{O}\left(\beta^{\frac{k}{2}}\right), \label{eq:ceorhigherord}
\end{equation}
showing that at the colder temperatures QuanTA permits higher order behaviour than the standard PT scheme which has $\hat{\ell}= \mathcal{O}\left(\beta\right)$.
\end{theorem}

\begin{proof}
Since the optimal $\ell$ derived in Theorem~\ref{Thr:optscalQuanTA} is given by \[\hat{\ell} \propto \left[\frac{1}{2} V(\beta) -I(\beta)+\frac{1}{4\beta}R(\beta) \right]^{-1/2}\] 
the proof of Theorem~\ref{cor:higherorder} follows immediately if it can be shown that
\begin{equation}
	\left[\frac{1}{2} V(\beta) -I(\beta)+\frac{1}{4\beta}R(\beta) \right]=\mathcal{O}\left(\frac{1}{\beta^k}\right). \label{neededcor}
\end{equation}

Indeed, two key Lemmata are derived in Section~\ref{sec:highprof} in the Appendix: Lemma~\ref{lem:highterm1} establishes that $\frac{1}{2} V(\beta) -I(\beta)=\mathcal{O}\left(\frac{1}{\beta^k}\right)$ and Lemma~\ref{lem:highterm2}  establishes that $\frac{1}{4\beta}R(\beta)=\mathcal{O}\left(\frac{1}{\beta^k}\right)$. Thus the result in \eqref{neededcor} holds and the proof is complete.

\end{proof}

\textit{Remark 4:} The result in Theorem~\ref{cor:higherorder} does not imply that QuanTA isn't useful outside the Gaussian or super cold settings. The QuanTA approach will be practically useful in settings where the mode can be well approximated by a Gaussian and thus allow the shift move to approximately preserve the quantile. What Theorem~\ref{cor:higherorder} does show is that for a large class of distributions that exhibit appropriate smoothness, QuanTA is sensible, and is arguably the canonical approach to take at the super cold levels, since it enables acceleration of the mixing speed through the temperature schedule.

\section{Examples of Implementation}
\label{subsec:Examples}

 This section gives illustrative examples for the canonical setting of a Gaussian mixture to illustrate the potential gains of QuanTA over the standard PT approach.

The QuanTA transformation move does not solve all the issues inherent in the PT framework. This will be highlighted with the final example in this section. In fact,  \cite{woodard2009conditions} and \cite{woodard2009sufficient} shows that for most ``interesting'' examples the mixing speed decays exponentially slowly with dimension. Prototype approaches to navigating this problem can be found in \cite{NTawnThesis}.

In each of the examples given, both the new QuanTA and standard (PT) parallel schemes will be run for comparison of performance. In all examples:
\begin{enumerate}
\item Both the new QuanTA and PT versions were run 10 times to ensure replicability.
\item Both the PT and QuanTA algorithms were run so that 20,000 swap moves would be attempted. For QuanTA this would be 20,000 swaps for each of the $N$ individual parallel tempering schemes in parallel of which there were $N=100$ in this example. Also all schemes had the same within to swap move ratio $(3:1)$.
\item Both versions use the same set of (geometrically generated) temperature spacings; chosen to be overly ambitious for the PT setup but demonstrably under-ambitious for the new QuanTA scheme.
\item Also presented is the optimal temperature schedule for the PT setup generated under the optimal acceptance rate of 0.234 for the PT algorithm suggested by \cite{atchade2011towards}. This demonstrates the extra complexity needed to produce a functioning algorithm for the PT approach.
\item For all runs, the within temperature level proposals were made with Gaussian RWM moves tuned to an optimal 0.234 acceptance rate, \cite{roberts1997weak}.
\end{enumerate}

\subsection{One-dimensional Example}
\label{reparam1Dex}

Target distribution given by:
\begin{equation}
	\pi(x) \propto \sum_{k=1}^5 w_k \phi(x;\mu_k,\sigma^2)
	\label{eq:onedimrepar5mode}
\end{equation}
where $\phi(\cdot;\mu,\sigma^2)$ is the density function of a univariate Gaussian with mean $\mu$ and variance $\sigma^2$. In this example, $\sigma=0.01$, the mode centres are given by $(\mu_1,\mu_2,\mu_3,\mu_4,\mu_5)=(-200,-100,0,100,200)$ and all modes are equally weighted with $w_1=w_2=\ldots=w_5$.

The temperature schedule for this example is given by a geometric schedule with an ambitious $0.0002$ common ratio for the spacings. Only 3 levels are used and so the temperature schedule is given by $\Delta=\{ 1,0.0002,0.0002^2 \}$, see Figure~\ref{Fig:onedfive}.
\begin{figure}[h]
\begin{center}
\includegraphics[keepaspectratio,width=9cm]{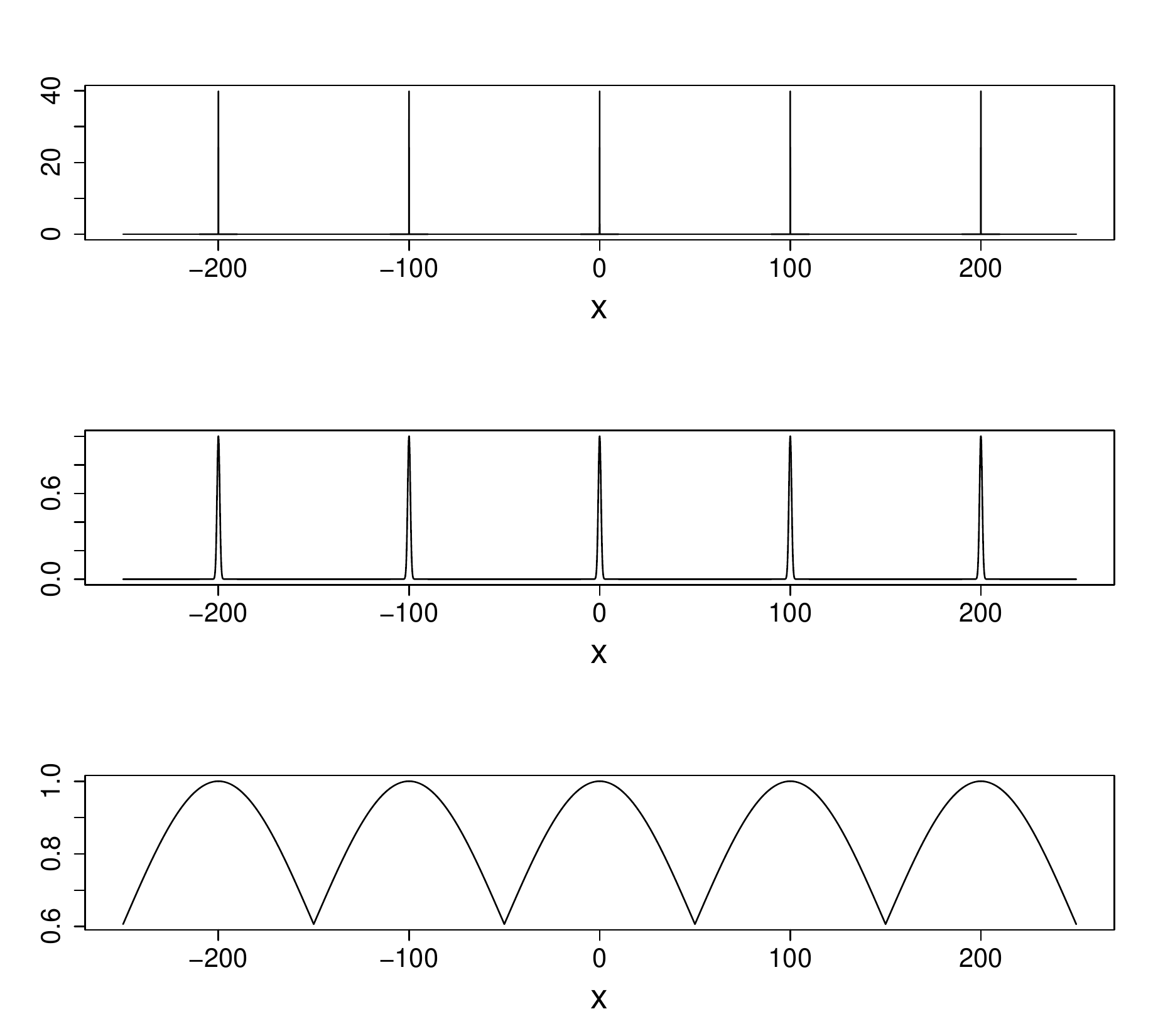}
\caption{The (non-normalised) tempered target distributions for (\ref{eq:onedimrepar5mode}) for inverse temperatures $\Delta=\{ 1,0.0002,0.0002^2 \}$ respectively.  }
\label{Fig:onedfive}
\end{center}
\end{figure}

In all runs all the chains were started from  a start location of -200. 
Figure~\ref{Fig:Onedimreparmex}, from the introductory section, shows two representative trace plots of the target state chain for a run of the PT algorithm and a single scheme from QuanTA respectively. There is a clear improvement in the inter-modal mixing for the QuanTA.

Table~\ref{tab:1dcomp} gives the associated acceptance rates. Clearly the rate of transfer of mixing information from the hot states to the cold state is significantly higher for QuanTA.
\begin{table}[h]
\begin{center}
\begin{tabular}{|c||c|c|}
  \hline
    Swap location: & 1& 2  \\ \hline \hline
  PT & 0.06 & 0.07  \\ \hline
  QuanTA & 0.99 & 0.99 \\
  \hline 
\end{tabular}
\caption{Comparison of the acceptance rates of swap moves for the PT algorithm and QuanTA targeting the one dimensional distribution given in (\ref{eq:onedimrepar5mode}) and setup with the ambitious inverse temperature schedule given by $\Delta=\{ 1,0.0002,0.0002^2 \}$.}
\label{tab:1dcomp}
\end{center}
\end{table}

Figure~\ref{Fig:OneDWEightconReparam} compares the running modal weight approximation for the mode centred on 200 when using the standard PT and QuanTA schemes respectively. This used the cold state chains from 10 individual runs of the PT algorithm and 10 single schemes selected randomly from 10 separate runs of the QuanTA algorithm. 

Denoting the estimator of the $k^{th}$ mode's weight by $\hat{w_k}$ and the respective cold state chain's $i^{th}$ value as $X_i$,
\begin{equation}
	\hat{w_k} = \frac{1}{N-B+1}\sum_{i=B}^N\mathbbm{1}_{\{ c_k<X_i\le C_k\}}.
	\label{eq:weightreparamest}
\end{equation}
where $c_k$ and $C_k$ are the chosen upper and lower boundary points for allocation to the $k^{th}$ mode; and $B$ is the length of the burn-in removed.

Figure~\ref{Fig:OneDWEightconReparam} shows the QuanTA approach has a vastly improved rate of convergence; with the PT runs still exhibiting bias from the chain initialisation locations.
\begin{figure}[h]
\begin{center}
\includegraphics[keepaspectratio,width=10cm]{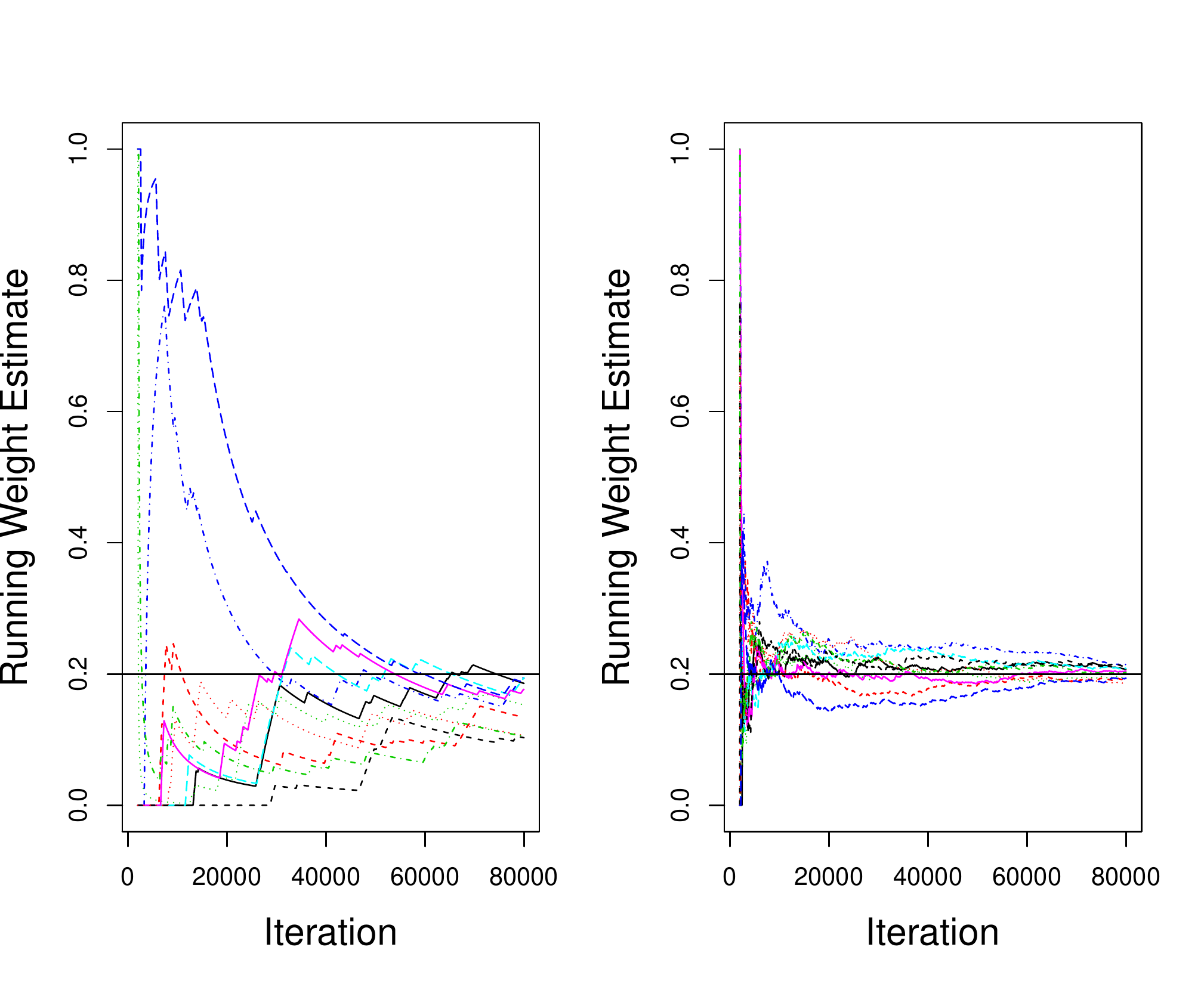}
\caption{For the target given in (\ref{eq:onedimrepar5mode}), the running weight approximations for the mode centred on 200 with target weight $w_5=0.2$ for 10 separate runs of the PT  and QuanTA schemes respectively. Left: the PT runs showing slow and variable estimates for $w_5$. Right: the new QuanTA scheme showing fast, unbiased convergence to the true value for $w_5$}
\label{Fig:OneDWEightconReparam}
\end{center}
\end{figure}

An interesting comparison between the approaches is to observe how many extra temperature levels would be required to make the PT scheme work optimally (i.e.\ with consecutive 0.234 swap acceptance rates). This gives a clearer idea of the reduction in number of intermediate levels that can be achieved using the QuanTA. 

With the same hottest state level of $\beta=0.0002^2$, a geometrical inverse temperature schedule was tuned to give a swap rate of approximately 0.234 was achieved between consecutive levels for the PT algorithm in this example. In fact a 0.04 geometric ratio suggested optimality for the PT scheme. Hence, to reach the stated hottest level needs 7 temperatures, as opposed to the 3 that were evidently unambitious for QuanTA. 

\subsection{Twenty-dimensional Example}
\label{reparam20Dex}

The target distribution is a 20-dimensional tri-modal Gaussian:
\begin{equation}
	\pi(x) \propto \sum_{k=1}^3 w_k \left[ \prod_{j=1}^{20} \phi(x_j;\mu_k,\sigma^2) \right].
	\label{eq:twentydimrepar3mode}
\end{equation}
In this example, $\sigma=0.01$, the marginal mode centres are given by $(\mu_1,\mu_2,\mu_3)=(-20,0,20)$ and all modes are equally weighted with $w_1=w_2=w_3$.The temperature schedule for this example is derived from a geometric schedule with an ambitious $0.002$ common ratio for the spacings. Only 4 levels are used and so the temperature schedule is given by $\{ 1,0.002,0.002^2,0.002^3 \}$.

\begin{figure}[h]
\begin{center}
\includegraphics[keepaspectratio,width=10cm]{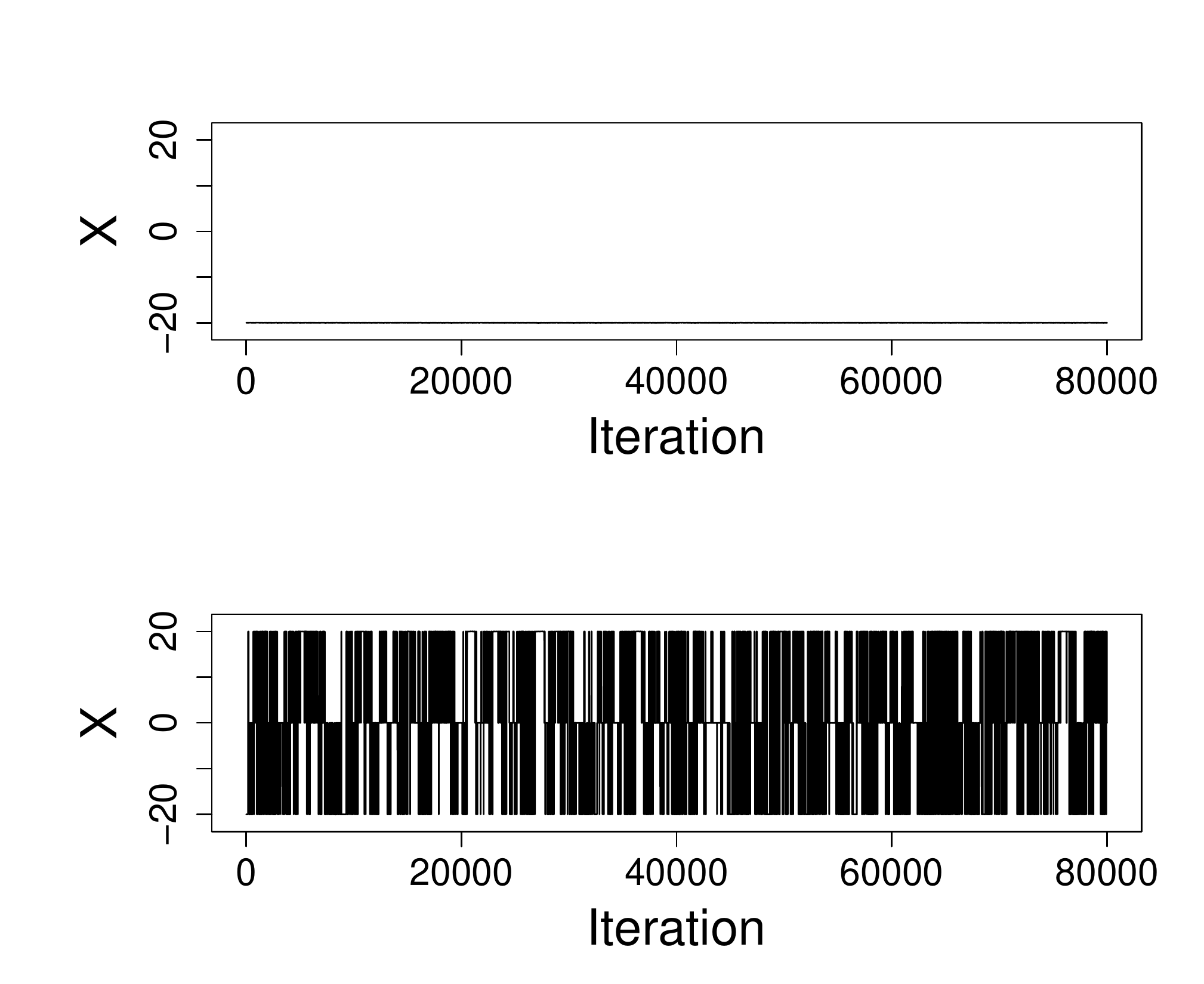}
\caption{Trace plots of the first component of the twenty dimensional cold state chains for representative runs of the PT (top) and new QuanTA (bottom) schemes. Note the fast inter-modal mixing of the new QuanTA scheme, allowing rapid exploration of the target distribution. In contrast the the PT scheme never escapes the initialising mode.}
\label{Fig:twentydimreparmex}
\end{center}
\end{figure}

In all runs all the chains were started from a  start location of $(-20,\ldots,-20)$.
Figure~\ref{Fig:twentydimreparmex} shows two representative trace plots of the target state chain for a run of the PT algorithm and QuanTA respectively. There is a clear improvement in the inter-modal mixing for the new QuanTA scheme. There is a stark contrast between the two algorithmic performances. The run using the standard PT scheme entirely fails to improve the mixing of the cold chain. In contrast the QuanTA scheme establishes a chain that is very effective at escaping the initialising mode and then mixes rapidly throughout the state space between the three modes.

The consecutive swap acceptance rates between the four levels are given in Table~\ref{tab:20dcomp}. Clearly there is no transfer of mixing information from the hot states to the cold state for the PT algorithm but there is excellent mixing in the QuanTA.
\begin{table}[h]
\begin{center}
\begin{tabular}{|c||c|c|c|}
  \hline
    Swap location: & 1& 2 & 3  \\ \hline \hline
  PT & 0 & 0 & 0  \\ \hline
  QuanTA & 0.99 & 0.99 & 0.99 \\
  \hline 
\end{tabular}
\caption{Comparison of the acceptance rates of swap moves for the PT algorithm and QuanTA targeting the Twenty dimensional distribution given in (\ref{eq:twentydimrepar3mode}) and setup with the ambitious inverse temperature schedule given by $\{ 1,0.002,0.002^2 ,0.002^3 \}$.}
\label{tab:20dcomp}
\end{center}
\end{table}

The temperature schedule choice that induces a 0.234 swap acceptance rate between consecutive temperature levels for this example using the PT algorithm indicates a geometric schedule with a 0.58 common ratio. This is in stark contrast to the 0.002 ratio that is evidently underambitious for QuanTA. Indeed, to reach the allocated hot state of $\beta=0.002^3$ then the PT algorithm would need 36 temperature levels in contrast to the 4 that sufficed for QuanTA.

\subsection{Five-dimensional Non-canonical Example}
\label{reparam5DexNonC}

Leaving the canonical symmetric mode setting, the following example has a five dimensional Gaussian mixture target with even weight to the modes but with different covariance scaling within each mode. The target distribution is given by:
\begin{equation}
	\pi(x) \propto \sum_{k=1}^3 w_k \left[ \prod_{j=1}^{5} \phi(x_j;\mu_k,\sigma_k^2) \right].
	\label{eq:Fivedimrepar3mode}
\end{equation}
 In this example, $(\sigma_1,\sigma_2,\sigma_3)=(0.02,0.01,0.015)$, the marginal mode centres are given by $(\mu_1,\mu_2,\mu_3)=(-20,0,20)$ and all modes are equally weighted with $w_1=w_2=w_3$.

Although at first glimpse this does not sound like a significantly harder problem, or even far from the canonical setting, the differing modal scalings make this a much more complex example. This is due to the lack of preservation of modal weight through power-based tempering, \cite{woodard2009sufficient}, with prototype solutions established in \cite{NTawnThesis}.

The temperature schedule for this example cannot be a simple geometric schedule as in the previous example due to the scaling indifference between the modes. By using an ambitious geometric schedule, the clustering was very unstable early on and this often led to an inability to establish mode centres for the run. Instead, a mixture of geometric schedules was used with an ambitious spacing for the coldest levels  and then a less ambitious spacing for the hotter levels. For the four coldest states an ambitious geometric schedule with $0.08$ common ratio was used. A further 8 hotter levels were added using a conservative  geometric schedule with ratio $0.4$. Hence the schedule was given by:
\begin{equation}
	\Delta=\{ 1,0.08,0.08^2,0.08^3,0.4^9, 0.4^{10},\ldots,0.4^{15},0.4^{16} \}.
	\label{eq:fiveschedule}
\end{equation}
 For the QuanTA scheme, the transformation moves were used for swap moves between the coldest 7 levels and standard swap moves were used otherwise.

Figure~\ref{Fig:FivedimREPARAMproblemplottrace} shows two representative trace plots of the target state chain for a run of the PT and QuanTA algorithms respectively. There is a clear improvement in the inter-modal mixing for the QuanTA scheme; albeit far less stark than that in the canonical one-dimensional and twenty-dimensional examples already shown. The run using the standard PT scheme fails to explore the state space. The QuanTA scheme establishes a chain that is able to explore the state space but does appear to have a bit of trouble during burn-in; mixing is good therein.
\begin{figure}[h]
\begin{center}
\includegraphics[keepaspectratio,width=10cm]{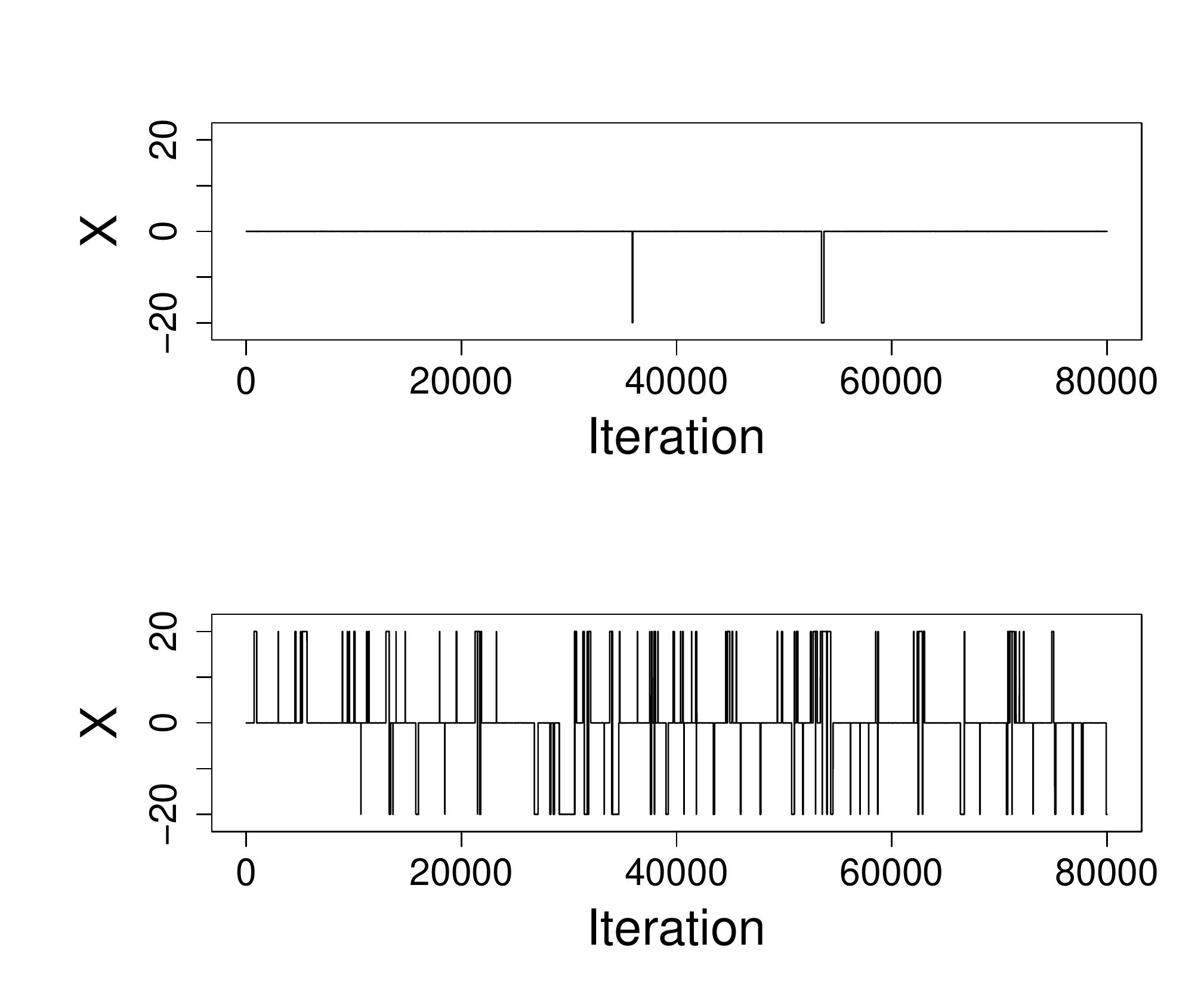}
\caption{Trace plots of the first component of the five dimensional cold state chains for representative runs of the PT and QuanTA schemes respectively. Note the difference in inter-modal mixing between the QuanTA scheme and the PT scheme which struggles to escape the initialisation mode.}
\label{Fig:FivedimREPARAMproblemplottrace}
\end{center}
\end{figure}

The consecutive swap acceptance rates between the 12 levels are given in Table~\ref{tab:20dcomp}. Clearly there is very poor mixing through the 4 coldest states for the PT algorithm. In contrast the QuanTA scheme has solid swap acceptance rates through the coldest levels but, unlike the previous examples, they are not all close to 1. 
\begin{table}[h]
\begin{center}
\begin{tabular}{|c||c|c|c|c|c|c|}
  \hline
  Swap location:  & 1& 2 & 3 &4&5 &6 \\ \hline \hline
  PT &  0.001 & 0.0161 & 0.0138 & 0.469 & 0.317 & 0.348   \\ \hline
  QuanTA & 0.446 & 0.970 & 0.997 & 0.999 & 0.999 & 0.999\\\hline \hline
	Swap location: &7&8&9&10&11 &- \\ \hline \hline
	PT  & 0.328 & 0.334 & 0.359 & 0.324 & 0.327&- \\ \hline
	QuanTA  & 0.285 & 0.285 & 0.285 & 0.285 & 0.302&-\\ \hline
\end{tabular}
\caption{Comparison of the acceptance rates of swap moves for the PT and new QuanTA algorithm targeting the five dimensional distribution given in (\ref{eq:Fivedimrepar3mode}) and setup with the ambitious inverse temperature schedule given in (\ref{eq:fiveschedule}). Note that for QuanTA, the reparametrised swap move was only used for swaps in the coldest 7 levels.}
\label{tab:5dcomp}
\end{center}
\end{table}

This example is both positive (showing the improved mixing using the QuanTA scheme on a hard example) but also serves as a warning for the degeneracy of both the PT and new QuanTA schemes when using power-based tempering on a target outside of the canonical symmetric mode setting.

\subsection{The Computational Cost of QuanTA}
\label{subsec:CompQuan}

It is important to analyse the computational cost of QuanTA. To be an effective algorithm the inferential gains of QuanTA per iteration should not be outweighed by the increase in run-time.

The analysis uses the runs of the one and twenty-dimensional examples, given above, using both the QuanTA and PT approaches. The  algorithms were  setup the same as in the ambitious versions of the spacing schedules in each case.

The key idea is to first establish the total run-time, denoted $R$, in each case. Typically one looks to compare the time-standardised Effective Sample Size (ESS). In this case it is natural to take the acceptance rate as a direct proxy for the effective sample size. This is due to the fact that the target distributions have symmetric modes with equal weights. Hence the acceptance rate between consecutive temperature levels  dictates the performance of the algorithm; in particular the quality of inter-modal mixing. 

To this end, taking the first level temperature swap acceptance rate, denoted $A$, the runs are compared using  run-time standardised acceptance rates i.e.\ $A/R$. 

Note that in both dimensional cases, the output from  QuanTA is 100 times larger due to the use of 100 schemes running in parallel. Hence, for a standardised comparison the time was divided by 100. Therefore, in what follows in this section, when the run-time, $R$, of the QuanTA approach is referred to, this means the full run-time divided by 100. The fairness of this is discussed below.

\begin{table}[h]
\begin{center}
\begin{tabular}{|c||c|c|}
  \hline
  Algorithm  & PT & QuanTA  \\ \hline \hline
  Run-time (sec) &  5.60 &  8.01   \\ \hline
  Swap Rate &  0.06  & 0.99   \\ \hline
	$A/R$ &  0.01 &   0.12 \\\hline
\end{tabular}
\caption{Complexity comparisons between QuanTA and PT for the one-dimensional example.}
\end{center}
\end{table}

\begin{table}[h]
\begin{center}
\begin{tabular}{|c||c|c|}
  \hline
  Algorithm  & PT & QuanTA  \\ \hline \hline
  Run-time (sec) &  8.00 &  12.79   \\ \hline
  Swap Rate &  0.00  & 0.99   \\ \hline
	$A/R$ &  0.00 &   0.08 \\\hline
\end{tabular}
\caption{Complexity comparisons between QuanTA and PT for the twenty-dimensional example.}
\end{center}
\end{table}

In both cases the QuanTA approach has a longer run-time to generate the same amount of output; as would be expected due to the added cost of clustering. Indeed, it takes approximately 1.5 times longer to generate the ``same amount of output''. 

However, the temperature swap move  acceptance rates are  16.5 and $\infty$ times better respectively when using the QuanTA approach. Using the acceptance rate as a proxy for  effective sample size then the quantity $A/R$ is the fundamental value for comparison. In both cases the QuanTA approach shows a significant improvement over the PT approach.

There are issues with the fairness of this comparison:
\begin{itemize}
  \item By standardising the run-time of QuanTA by the number of parallel schemes is not fully fair since it is sharing out the clustering expense between schemes. 
	\item The spacings are too ambitious for the PT approach meaning that the acceptance rates are very low. For a complete analysis one should run the PT algorithm on its optimal temperature schedule and then use the time-standardised effective sample size from each of the optimised algorithms.
\end{itemize}

The empirical computational studies are favourable to the QuanTA approach. This is for a couple of examples that are canonical for QuanTA. Outside of this canonical setting the improvements from running QuanTA will be less obvious.

\section{Conclusion and Further work}

The prototype QuanTA approach utilises a non-centred transformation approach to accelerate the  transfer of mixing information from the rapidly mixing ``hot'' state to aid the inter-modal mixing in the target ``cold'' state. Examples show that this novel algorithm has the potential to dramatically improve the inferential gains; particularly in settings where the modes are similar to a Gaussian in structure.

The accompanying theoretical results that are given in Section~\ref{sec:optscale} show that in a generic non-Gaussian setting the QuanTA approach can still exhibit accelerated mixing through the temperature schedule. Although the inverse temperature spacings are generally still $\mathcal{O}(d^{-1/2})$ there is a higher order behaviour exhibitted in the mixing for large (i.e.\ cold) values of the inverse temperature $\beta$. This suggests that the QuanTA approach will be powerful for accelerating the mixing through the colder levels of the temperature schedule for a typical smooth target.

It is clear that there are interesting questions to be addressed and further work needed before  QuanTA can be considered practical in a real data problem. In terms of optimising the computational expense, it has been shown that parallelisation of the PT algorithm can give significant practical gains, \cite{VanDerwerken2013};  by design QuanTA is also highly parallelisible. An impracticality of the current clustering method used is that is requires prior specification of the number of modes, $K$ which is likely to be unknown and would need  online-estimation as part of the clustering process. An interesting question is whether using colder levels along with the weighted clustering would help to aid the stability of the clustering once invariance is reached for the population. Indeed, the mixing at these auxiliary super cold levels should be very fast due to the QuanTA exhibiting higher order behaviour in these modes. The other interesting question is regarding the robustness of the method in heavier tailed modes, when the Gaussian approximation to the mode can be poor. Consider the setting of a univariate Laplace distribution and observe that the QuanTA style transformation never agrees with the ideal CDF preserving transformation. Some initial ideas and details of this further work can be found in \cite{NTawnThesis}.

\section{Appendix}
\label{sec:proofQuanTA}

This section gives the proof details of the results in Section~\ref{sec:scalingstheorem}. Firstly, some key notation is introduced that will be useful throughout this section.

\begin{definition} \label{def:Keynotations}
Denote:
\begin{itemize}
\item $ B= \log \left( \frac{f_d^{\beta'}(g(x,\beta,\beta^{'}))f_d^{\beta}(g(y,\beta^{'},\beta))}{f_d^{\beta'}(y)f_d^{\beta}(x)}  \right)$;
\item $h(x) := \log \left( f(x) \right)$;
\item $k(x):=(x-\mu)h'(x)$; 
\item and $r(x):=(x-\mu)^2h''(x).$
\end{itemize}
Then define
\begin{eqnarray}
M(\beta)&=&\mathbb{E}_{\beta}(h(z)) \\
S(\beta)&=&\mathbb{E}_{\beta}(k(z))  \\
R(\beta)&=&\mathbb{E}_{\beta}(r(z)-k(z)), \label{eq:Rbetadef}
\end{eqnarray}
where all expectations are with respect to the distribution $\frac{f^\beta (x)}{Z_{\beta}}$ where $Z_{\beta}=\int f^\beta (z) dz$.
\end{definition}


\begin{proposition}
Under the notation and assumptions of Theorem~\ref{Thr:optscalQuanTA} and definition~\ref{def:Keynotations} then it can be shown that
\begin{eqnarray}
I(\beta):=M'(\beta)
=\mathrm{Var}_\beta(h(x)). \label{eq:Ibeta}
\end{eqnarray}
and
\begin{eqnarray}
S(\beta)  =-\frac{1}{\beta}, \label{eq:Sbeta}
\end{eqnarray}
which trivially gives that
 \begin{equation}
V(\beta):=S'(\beta)= \frac{1}{\beta^2}. \label{eq:vbeta}
\end{equation}
\end{proposition}

\begin{proof}
The proof of \eqref{eq:Ibeta} is routine and can be found in \cite{atchade2011towards}. The derivation of \eqref{eq:Sbeta} is less obvious using integration by parts:
\begin{eqnarray*}
S(\beta)&=&\int (x-\mu)  (\log f)'(x) \frac{f^\beta (x)}{Z_{\beta}}dx \\
&=&  \int (x-\mu)  f'(x) \frac{f^{\beta-1} (x)}{Z_{\beta}}dx \\
&=& \cancelto{~0}{\left[ \frac{(x-\mu)}{\beta} \frac{f^{\beta} (x)}{Z_{\beta}}  \right]_{-\infty}^{-\infty}}-\int \frac{1}{\beta} \frac{f^{\beta} (x)}{Z_{\beta}}dx 
\end{eqnarray*}
\end{proof}



\subsection{Proof of Theorem~\ref{Thr:optscalQuanTA} }

This section derives 3 key results that are specific to deriving the result in Theorem~\ref{Thr:optscalQuanTA}. Lemma~\ref{lemm:Taylorexpand} will establish a Taylor expanded form of the log acceptance ratio of a temperature swap move that will prove to be asymptotically useful. Lemma~\ref{lem:asymgaussiy} will then establish the limiting Gaussianity of this logged acceptance ratio and finally, Lemma~\ref{lem:ESJDOPT} completes the proof of Theorem~\ref{Thr:optscalQuanTA} by establishing the optimal spacings and associated optimal acceptance rates required.

\begin{lemma}[QuanTA Log-Acceptance Ratio] \label{lemm:Taylorexpand}
Under the notation and assumptions of Theorem~\ref{Thr:optscalQuanTA} and definition~\ref{def:Keynotations},
\begin{eqnarray}
B &=& \epsilon\left[\sum_{i=1}^d
h(x_i)- h(y_i)+\frac{1}{2} \left( k(y_i)-k(x_i)\right)\right]\nonumber\\
&& +\frac{\epsilon^2}{8\beta} \left[\sum_{i=1}^d r(x_i)-k(x_i)+r(y_i)-k(y_i)\right]+(T_x +T_y).
\end{eqnarray}
where both $T_x\rightarrow 0$ and $T_y\rightarrow 0$ in probability as $d\rightarrow\infty.$
\end{lemma}

\begin{proof}
By taking logarithms it is immediate that
\begin{eqnarray}
B &=&  \sum_{i=1}^d [\beta' h (g(x_i,\beta,\beta^{'})) - \beta h(x_i)]+ \sum_{i=1}^d[\beta h(g(y_i,\beta^{'},\beta)) -\beta' h (y_i) ] \nonumber \\
&=:& H_\beta^{\beta'} (\mathbf{x}) +H_{\beta'}^{\beta} (\mathbf{y}) \label{eq:Bterms}.
\end{eqnarray}

With the aim being to derive the asymptotic behaviour of the log acceptance ratio then the next step is to use Taylor expansions (in $\epsilon$) to appropriate order so that the asymptotic behaviour of $B$ can be understood.

 For notational convenience, the following will be used:
\begin{itemize}
	\item  Making $h(g(x,\beta,\beta^{'}))$  explicitly dependent on $\epsilon$ \[ \alpha_x(\epsilon):= h(g(x,\beta,\beta^{'}))= \log \left[  f \left(  \left(\frac{\beta}{\beta+\epsilon} \right)^{1/2} (x-\mu) +\mu\right) \right]. \] 
	\item  Denote \[  d_x(\epsilon):= \left(\frac{\beta}{\beta+\epsilon} \right)^{1/2} (x-\mu) +\mu. \] 
\end{itemize}

By Taylor series expansion in $\epsilon$, for  fixed $x$, with Taylor remainder correction term denoted by $\xi_x$ such that $0<\xi_x<\epsilon$:
\begin{eqnarray}
~~~~~h(g(x,\beta,\beta^{'}))=\alpha_x(\epsilon) = \alpha_x(0)+\epsilon\alpha_x'(0)+\frac{\epsilon^2}{2}\alpha_x''(0)+\frac{\epsilon^3}{6}\alpha_x'''(\xi_x), \label{eq:expanx}
\end{eqnarray}
where
\begin{eqnarray}
\alpha_x'(\epsilon)&=& -\frac{(x-\mu)}{2}\frac{\beta^{1/2}}{(\beta+\epsilon)^{3/2}}(\log f)'(d_x(\epsilon)), \label{eq:rigor1}\\
\alpha_x''(\epsilon)&=& \frac{(x-\mu)^2}{4}\frac{\beta}{(\beta+\epsilon)^{3}}(\log f)''(d_x(\epsilon)) \label{eq:rigor2}\\
&&+\frac{3(x-\mu)}{4}\frac{\beta^{1/2}}{(\beta+\epsilon)^{5/2}}(\log f)'(d_x(\epsilon)),  \nonumber\\
\alpha_x'''(\epsilon)&=& -\frac{(x-\mu)^3}{8}\frac{\beta^{3/2}}{(\beta+\epsilon)^{9/2}}(\log f)'''(d_x(\epsilon)) \label{eq:rigor3}\\
&&-\frac{9(x-\mu)^2}{8}\frac{\beta}{(\beta+\epsilon)^{4}}(\log f)''(d_x(\epsilon)) \nonumber\\
&&-\frac{15(x-\mu)}{8}\frac{\beta^{1/2}}{(\beta+\epsilon)^{7/2}}(\log f)'(d_x(\epsilon)).\nonumber
\end{eqnarray}
As a preview to the later stages of this proof,  the terms up to second order in $\epsilon$ dictate the  asymptotic distribution of $B$. However, to show that the higher order terms ``disappear'' in the limit as $\epsilon \rightarrow 0$ then a careful analysis is required.  Thus the  next step is to establish that, under the assumptions made above, the higher order terms converge to zero in probability.

To this end, a careful analysis of $\alpha_x'''(\cdot)$ is undertaken. Firstly, it will be shown that $\left|\mathbb{E}_\beta[\alpha_x'''(\xi_x)] \right|$ is bounded; then application of Markov's inequality will establish that the higher order terms converge to zero in probability as $d \rightarrow \infty$.  Define
\[ \eta_\epsilon := \left[  \left( \frac{\beta}{\beta+\epsilon} \right) ^ \frac{1}{2}-1 \right]   \]
so that
\[   d_x(\epsilon) -x =\left[  \left( \frac{\beta}{\beta+\epsilon} \right) ^ \frac{1}{2}-1 \right] (x-\mu) :=\eta_\epsilon (x-\mu),\]
which has the property that $\eta_\epsilon\rightarrow 0 ~~\mbox{as}~~ d\rightarrow \infty$ and 
 $|\eta_\epsilon| \le 1$.

Then, with Taylor remainder correction terms denoted $\xi^\epsilon_1,\xi^\epsilon_2,\xi^\epsilon_3$ such that $0<|\xi^\epsilon_k-x|<|   d_x(\epsilon) -x |$
\begin{eqnarray}
(\log f)'(d_x(\epsilon)) &=& (\log f)'(x)+  \eta_\epsilon (x-\mu)(\log f)''(x)   \label{eq:rigorx1}\\
&&+\frac{ {\eta_\epsilon}^2(x-\mu)^2}{2}(\log f)'''(x) \nonumber\\
&& + \frac{ {\eta_\epsilon}^3(x-\mu)^3}{6}(\log f)''''(\xi^\epsilon_1),\nonumber \\
(\log f)''(d_x(\epsilon)) &=& (\log f)''(x)+  \eta_\epsilon (x-\mu) (\log f)'''(x)\label{eq:rigorx2}\\
&& +\frac{ {\eta_\epsilon}^2(x-\mu)^2}{2}(\log f)''''(\xi^\epsilon_2), \nonumber\\
(\log f)'''(d_x(\epsilon)) &=& (\log f)'''(x)+  \eta_\epsilon (x-\mu) (\log f)''''(\xi^\epsilon_3). \label{eq:rigorx3}
\end{eqnarray}
Recall the assumptions \eqref{eq:reguvar} and \eqref{eq:quantafourthass}. Substituting \eqref{eq:rigorx1}, \eqref{eq:rigorx2} and \eqref{eq:rigorx3} into \eqref{eq:rigor3}; evaluating the expectation with respect to $X\sim f^\beta$ and for convenience denoting $|x-\mu|$ by $S$  then $\exists~C \in \mathbb{R}_+$
\begin{eqnarray}
\left|\mathbb{E}_\beta[\alpha_x'''(\xi_x)] \right|
	&\le& 	\mathbb{E}_\beta \left[\left|\alpha_x'''(\xi_x)\right|\right]  \nonumber \\
&\le& \mathbb{E}_\beta\Bigg[\frac{S^3}{8}\beta^{-3} |(\log f)'''(d(\xi_x))| \nonumber \\&& +\frac{9S^2}{8} \beta^{-3} |(\log f)''(d(\xi_x))|+\frac{15 S}{8} \beta^{-3}|(\log f)'(d(\xi_x))|\Bigg] \nonumber \\
&\le& \mathbb{E}_\beta\Bigg[ \frac{S^3}{8}\beta^{-3} \left(|(\log f)'''(x)|+S|(\log f)''''(\xi^{\xi_x}_3)|\right) \nonumber\\
&& +\frac{9S^2}{8} \beta^{-3} \Bigg(|(\log f)''(x)|+S|(\log f)'''(x)|\nonumber \\
&&+\frac{|x|^2}{2}|(\log f)''''(\xi^{\xi_x}_2)|\Bigg) +\frac{15 S}{8} \beta^{-3}\Bigg(|(\log f)'(x)|\nonumber \\ 
&&+S|(\log f)''(x)|+\frac{S^2}{2}|(\log f)'''(x)|\nonumber \\
&&+\frac{S^3}{6}|(\log f)''''(\xi^{\xi_x}_1)|\Bigg) \Bigg]
\le C \label{eq:alphaboundx}
\end{eqnarray}
where the first three inequalities are from the direct application of the triangle inequality (with the second also using the boundedness of  $\eta_\epsilon$); whereas the final inequality arises from both the finiteness of expectations of the terms involving derivatives of order three or below (this is due to the regularly varying tails of $\log (f(\cdot))$) and the assumption that $|(\log f)''''(\cdot)|<M$.

Using \eqref{eq:expanx}, with substitution of terms from \eqref{eq:rigor1}, \eqref{eq:rigor2} and \eqref{eq:rigor3},  $H_\beta^{\beta'} (\mathbf{x})$  can be expressed as
\begin{eqnarray}
H_\beta^{\beta'} (\mathbf{x})  &=& \sum_{i=1}^d(\beta+\epsilon) \left[ \alpha_{x_i}(\epsilon) - \beta  \alpha_{x_i}(0) \right] \nonumber \\
&=&\epsilon \sum_{i=1}^d \left[\alpha_{x_i}(0)+\beta \alpha_{x_i}'(0)\right] + \epsilon^2 \sum_{i=1}^d \left[\frac{\beta}{2}\alpha_{x_i}''(0)+ \alpha_{x_i}'(0)\right]\nonumber\\
&&~~+\epsilon^3 \sum_{i=1}^d \left[\frac{1}{2}\alpha_{x_i}''(0)+\frac{\beta}{6}\alpha_{x_i}'''(\xi_{x_i})\right] + \epsilon^4 \sum_{i=1}^d \frac{1}{6} \alpha_{x_i}'''(\xi_{x_i}).
\end{eqnarray}
By \eqref{eq:alphaboundx} and using the iid nature of the $x_i's$ and using Markov's inequality then $\forall \delta>0$
\begin{eqnarray}
\delta \mathbb{P}\Bigg(  \Bigg| \epsilon^3 \sum_{i=1}^d \Bigg[\frac{1}{2}\alpha_{x_i}''(0) &+&\frac{\beta}{6}\alpha_{x_i}'''(\xi_{x_i})\Bigg]  \Bigg|>\delta \Bigg)\nonumber\\
&<&\mathbb{E}\left(  \left| \frac{\ell^3}{d^{3/2}} \sum_{i=1}^d \left[\frac{1}{2}\alpha_{x_i}''(0)+\frac{\beta}{6}\alpha_{x_i}'''(\xi_{x_i})\right]  \right| \right) \nonumber\\
&\le& \frac{\ell^3}{d^{1/2}}\left[ \frac{1}{2}\mathbb{E}\left(  |\alpha_{x_i}''(0)| \right)+ \frac{\beta}{6}C\right] \rightarrow 0 ~~\mbox{as}~~ d\rightarrow \infty.\nonumber
\end{eqnarray}
Thus, \[ \epsilon^3 \sum_{i=1}^d \left[\frac{1}{2}\alpha_{x_i}''(0)+\frac{\beta}{6}\alpha_{x_i}'''(\xi_{x_i})\right]\rightarrow 0 ~~\mbox{in probability as}~~ d\rightarrow \infty.\]
 By identical methodology, as $d\rightarrow \infty$ \[  \epsilon^4 \sum_{i=1}^d \frac{1}{6} \alpha_{x_i}'''(\xi_{x_i}) \rightarrow 0 ~~\mbox{in probability}.\]
Consequently,
\begin{eqnarray}
H_\beta^{\beta'} (\mathbf{x}) &=& \epsilon \left[\sum_{i=1}^d h(x_i)-\frac{1}{2}(x_i-\mu)h'(x_i)\right] \nonumber\\
&& +\frac{\epsilon^2}{8\beta} \left[\sum_{i=1}^d (x_i-\mu)^2h''(x_i)-(x_i-\mu)h'(x_i)\right] + T_x \label{eq:keyx}
\end{eqnarray}
where 
\[ T_x = \epsilon^3 \sum_{i=1}^d \left[\frac{1}{2}\alpha_{x_i}''(0)+\frac{\beta}{6}\alpha_{x_i}'''(\xi_{x_i})\right] + \epsilon^4 \sum_{i=1}^d \frac{1}{6} \alpha_{x_i}'''(\xi_{x_i})\]
with $T_x \rightarrow 0$ in probability as $d\rightarrow \infty$.

Now denoting $h(g(y,\beta^{'},\beta))$ as
\[   \alpha_y(\epsilon):=h(g(y,\beta^{'},\beta))=  \log \left[  f \left(  \left(\frac{\beta+\epsilon}{\beta} \right)^{1/2} (y-\mu)+\mu \right) \right],  \]
the Taylor series expansion in $\epsilon$, for  a fixed $y$, with Taylor truncation term denoted by $\xi_y$ such that $0<\xi_y<\epsilon$ is given by
\begin{eqnarray}
~~~h(g(y,\beta^{'},\beta))=\alpha_y(\epsilon) = \alpha_y(0)+\epsilon\alpha_y'(0)+\frac{\epsilon^2}{2}\alpha_y''(0)+\frac{\epsilon^3}{6}\alpha_y(\xi_y). \label{eq:expany}
\end{eqnarray}
By identical methodology to the above calculation in \eqref{eq:alphaboundx} for $\alpha_x(\cdot)$, it can be shown that $\exists ~C_y\in \mathbb{R}_+$ such that
\begin{eqnarray}
\left|\mathbb{E}_\beta[\alpha_y'''(\xi_y)] \right| &\le& C_y. \label{eq:alphaboundy}
\end{eqnarray}
Hence, using exactly the same methodology as for the $x_i$'s above, then 
\begin{eqnarray}
H_{\beta'}^{\beta} (\mathbf{y})  
&=& -\epsilon \left[\sum_{i=1}^d h(y_i)-\frac{1}{2}(y_i-\mu)h'(y_i)\right] \nonumber\\
&& +\frac{\epsilon^2}{8\beta} \left[\sum_{i=1}^d (y_i-\mu)^2h''(y_i)-(y_i-\mu)h'(y_i)\right] + T_y. \label{eq:keyy}
\end{eqnarray}
where $T_y \rightarrow 0$ in probability as $d\rightarrow \infty$.

Using the notation from Definition~\ref{def:Keynotations} the desired form of $B$ in Lemma~\ref{lemm:Taylorexpand} is reached.
\end{proof}

\begin{lemma}[Asymptotic Gaussianity of the Log-Acceptance Ratio for QuanTA] \label{lem:asymgaussiy}
Under the notation and assumptions of Theorem~\ref{Thr:optscalQuanTA} and Definition~\ref{def:Keynotations}, $B$  is asymptotically Gaussian of the form  $B\dot{\sim} N(\frac{-\sigma^2}{2},\sigma^2)$ where 
\[\sigma^2= 2\ell^2\left[\frac{1}{2} V(\beta) -I(\beta)+\frac{1}{4\beta}R(\beta)\right]. \]
\end{lemma}

\begin{proof}

Recall the form of $B$ from Lemma~\ref{lemm:Taylorexpand}, then making the dimensionality dependence explicit, write $B=W(d)+(T_x+T_y)$ where
\begin{eqnarray}
W(d) &:=& \epsilon\left[\sum_{i=1}^d
h(x_i)- h(y_i)+\frac{1}{2} \left( k(y_j)-k(x_j)\right)\right]\nonumber\\
&& +\frac{\epsilon^2}{8\beta} \left[\sum_{i=1}^d r(x_i)-k(x_i)+r(y_i)-k(y_i)\right] \nonumber
\end{eqnarray}
and $(T_x+T_y) \rightarrow 0$ in probability as $d \rightarrow \infty$. Hence, if it can be shown that $W(d)$ converges in distribution to a Gaussian of the form $N(-c, 2c)$ then by Slutsky's Theorem one can conclude that $B$ converges in distribution to the same Gaussian as the $W$.

To this end, the asymptotic Gaussianity of $W(d)$ is established. First note that due to the iid nature of the $x_i$'s and $y_i$'s respectively then by the standard Central Limit Theorem, e.g. \cite{Durrett2010}, for a sum of iid variables, then asymptotic Gaussianity is immediate where
\begin{equation}
	W(d) \Rightarrow N\left( \mu_W, \sigma^2_W\right)~~\mbox{as}~~d\rightarrow \infty
\end{equation}
where
\[ \mu_W = \lim _{d\rightarrow \infty} \mathbb{E}[W(d)]~~\mbox{and}~~ \sigma^2_W = \lim _{d\rightarrow \infty} \mbox{Var}[W(d)].\] To this end the terms $\mathbb{E}[W(d)]$ and $\mbox{Var}[W(d)]$ are computed.
\begin{eqnarray}
\mathbb{E}[W(d)] &:=& \epsilon\left[\sum_{i=1}^d
M(\beta)-M(\beta+\epsilon) -\frac{1}{2}(S(\beta)-S(\beta+\epsilon))\right]\nonumber\\
&& +\frac{\epsilon^2}{8\beta} \left[\sum_{i=1}^d R(\beta)+R(\beta+\epsilon)\right] \nonumber \\
&=& \epsilon\left[\sum_{i=1}^d
-\epsilon M'(\beta)+\frac{\epsilon}{2}S'(\beta)\right]
 +\frac{\epsilon^2}{8\beta} \left[\sum_{i=1}^d 2R(\beta)\right] +\mathcal{O}(d^{-1/2})\nonumber\\
&&\rightarrow \ell^2\left[\frac{1}{2} V(\beta) -I(\beta)+\frac{1}{4\beta}R(\beta)\right] ~~\mbox{as}~~d \rightarrow \infty. \nonumber
\end{eqnarray}
Similarly,
\begin{equation}
	\mbox{Var}(W(d)) \rightarrow 2 \ell^2  \mathrm{Var}_{\beta}  \left( h(x)-\frac{1}{2}k(x) \right)~~\mbox{as}~~d \rightarrow \infty. \nonumber
\end{equation}
Hence by Slutsky's Theorem then  $B$ is asymptotically Gaussian such that
\begin{equation}
B~~ \dot{\sim}~~ N\left(\ell^2\left[\frac{1}{2} V(\beta) -I(\beta)+\frac{1}{4\beta}R(\beta)\right], 2 \ell^2  \mathrm{Var}_{\beta}  \left( h(x)-\frac{1}{2}k(x) \right)\right). \label{eq:limdis}
\end{equation}
However, this does not obviously have the form required with $B\dot{\sim} N(\frac{-\sigma^2}{2},\sigma^2)$ for some $\sigma^2$. This form is verified with the following Proposition~\ref{cor:minushalsig} and this then completes the proof of Lemma~\ref{lem:asymgaussiy}.

\begin{proposition}
\label{cor:minushalsig}
Under the notation and assumptions of Theorem~\ref{Thr:optscalQuanTA} and Definition~\ref{def:Keynotations} then
\begin{equation}
\ell^2\left[\frac{1}{2} V(\beta) -I(\beta)+\frac{1}{4\beta}R(\beta)\right]=- \ell^2  \mathrm{Var}_{\beta}  \left( h(x)-\frac{1}{2}k(x) \right). \label{coreq:identity}
\end{equation}
\end{proposition}

\begin{proof}
From (\ref{eq:limdis}) then denote
\begin{equation}
	\mu= \ell^2\left[\frac{1}{2} V(\beta) -I(\beta)+\frac{1}{4\beta}R(\beta)\right] \label{eq:muterm1}
\end{equation}
and
\begin{equation}
	\sigma^2= 2\ell^2  \mathrm{Var}_{\beta}  \left( h(x)-\frac{1}{2}k(x) \right). \nonumber
\end{equation}
Then by using the standard properties of variance it is routine to show that
\begin{equation}
	-\frac{\sigma^2}{2}= \ell^2  \left[ -I(\beta) -\frac{1}{4}  \mathrm{Var}_{\beta}(k(x)) + V(\beta)\right]. \label{eq:altvarterm}
\end{equation}
Consequently, equating the terms on the RHS of \eqref{eq:muterm1} and \eqref{eq:altvarterm} shows that if the following can be shown to hold then the required identity in \eqref{coreq:identity} is validated:
\begin{equation}
	\frac{1}{4\beta}R(\beta)= -\frac{1}{4}  \mbox{var}_{\beta}(k(x)) + \frac{1}{2}V(\beta). \label{eq:needequal}
\end{equation}
The LHS and RHS of \eqref{eq:needequal} will be considered separately. The following integration by parts are well defined due to the assumption that $-\log (f(\cdot))$ has regularly varying tails. Starting with the RHS and recalling that from \eqref{eq:Sbeta} $\mathbb{E}_{\beta}(k(x))=-1/\beta$:
\begin{eqnarray}
-\frac{1}{4}  \mbox{var}_{\beta}(k(x)) + \frac{1}{2}V(\beta)&=&-\frac{1}{4} \left[\mathbb{E}_{\beta}(k(x)^2)-\mathbb{E}_{\beta}(k(x))^2\right] + \frac{1}{2\beta^2} \nonumber\\
&=& -\frac{1}{4} \mathbb{E}_{\beta}(k(x)^2)+ \frac{3}{4\beta^2}. \nonumber
\end{eqnarray}
Then, noting that $(\log f)'(x) f^\beta(x)=f'(x)f^{\beta-1}(x)$, and using integration by parts (by first integrating $f'(x)f^{\beta-1}(x)$):
\begin{eqnarray}
\hspace{-1cm}\mathbb{E}_{\beta}(k(x)^2)&=&\int (x-\mu)^2 [(\log f)'(x)]^2 \frac{f^{\beta}(x)}{Z_{\beta}} dx \nonumber \\
&=& \cancelto{~0}{\left[ \frac{(x-\mu)^2}{\beta} (\log f)'(x)\frac{f^{\beta} (x)}{Z_{\beta}}  \right]_{-\infty}^{-\infty}}\nonumber \\
&&-\frac{1}{\beta}  \int \left[(x-\mu)^2 (\log f)''(x)+2(x-\mu) (\log f)'(x)\right]\frac{f^{\beta} (x)}{Z_{\beta}}dx \nonumber\\
&=& -\frac{1}{\beta}\mathbb{E}_{\beta}(r(x))-\frac{2}{\beta} \mathbb{E}_{\beta}(k(x)) = -\frac{1}{\beta}\mathbb{E}_{\beta}(r(x))+\frac{2}{\beta^2}. \label{eq:ek2}
\end{eqnarray}
Collating the above in \eqref{eq:needequal} and \eqref{eq:ek2} then
\begin{equation}
-\frac{1}{4}  \mbox{var}_{\beta}(k(x)) + \frac{1}{2}V(\beta)= \frac{1}{4\beta}\mathbb{E}_{\beta}(r(x))+ \frac{1}{4\beta^2} =\frac{1}{4\beta}R(\beta),\label{eq:idenRHS}
\end{equation}
 where the final equality simply comes from the definition of $R(\beta)$ from\eqref{eq:Rbetadef}.
\end{proof}
\end{proof}

%

\begin{lemma}[Optimisation of the $ESJD_\beta$] \label{lem:ESJDOPT}
Under the notation and assumptions of Theorem~\ref{Thr:optscalQuanTA} and Definition~\ref{def:Keynotations} then the $ESJD_{\beta}$,  is maximised when $\ell$ is chosen to maximise
\begin{equation}
2\ell^2\Phi\left(-\frac{\ell\left[\frac{1}{2} V(\beta) -I(\beta)+\frac{1}{4\beta}R(\beta)\right]^{1/2}}{\sqrt{2}}\right), \nonumber
\end{equation}

Furthermore, for the optimal $\ell$ the corresponding swap move acceptance rate induced between two consecutive temperatures is given by 0.234 (3.s.f).
\end{lemma}

\begin{proof}

Letting $\phi_{(m,\sigma^2)}$ denote the density function of a Gaussian with mean $m$ and variance $\sigma^2$ and suppose that $G\sim N(-\frac{\sigma^2}{2},\sigma^2)$ then a routine calculation (which can be found in e.g.\ \cite{roberts1997weak}) shows that
\begin{eqnarray}
	\mathbb{E}(1\wedge e^G) 
	= 2\Phi\left(-\frac{\sigma}{2}\right). \label{eq:expaccrateforGauus}
\end{eqnarray}
Using the result in \eqref{eq:expaccrateforGauus} and Lemma~\ref{lem:asymgaussiy}, then in the limit as $d \rightarrow \infty$ 
\begin{eqnarray}
\lim_{d\rightarrow \infty} (d~ESJD_\beta) = 2\ell^2\Phi\left(-\frac{\ell\left[\frac{1}{2} V(\beta) -I(\beta)+\frac{1}{4\beta}R(\beta)\right]^{1/2}}{\sqrt{2}}\right). \label{eq:ESJDphiforml}
\end{eqnarray}
Substituting
\begin{equation}
u=\ell \left[\frac{1}{2} V(\beta) -I(\beta)+\frac{1}{4\beta}R(\beta) \right]^{1/2}, \nonumber
\end{equation}
and then maximising with respect to $u$  attains an optimising value $u^*$ that doesn't depend on\[\left[\frac{1}{2} V(\beta) -I(\beta)+\frac{1}{4\beta}R(\beta) \right].\] 

Recalling the form of the $ESJD_\beta$ from \eqref{eq:ESJDP12form} , then it is clear that the associated acceptance rate, denoted $(ACC_\beta)$, induced by choosing the any value of $u$ is 
\[ACC_\beta = \mathbb{E}_{\pi_n} \left[ \alpha_{\beta}(x,y) \right] \] which, as established above, in the limit as $d \rightarrow \infty$ is asymptotically given by 
\[ ACC_\beta = 2 \Phi\left(-\frac{ u}{\sqrt{2}}\right) \]
Now it can be shown numerically that for the optimising value  $u^*$ induces
\begin{equation}
		\mbox{ACC}_\beta = 0.234 ~~~(3.s.f). \nonumber
\end{equation}
\end{proof}
%
%
%

\subsection{Lemmata for the Proof of Theorem~\ref{cor:higherorder}} \label{sec:highprof}

Note that in Theorem~\ref{cor:higherorder}, the conditions on $f(\cdot)$ are inherited from the conditions on $f(\cdot)$ from Theorem~\ref{Thr:optscalQuanTA}. This includes the bounded fourth derivatives of $\log(f)$ and the existence of eighth  moments i.e.\ $\mathbb{E}_\beta \left[ X^8 \right]$, which is due to the assumption of regularly varying tails. These will be assumed for the following lemmata.

\begin{lemma}\label{lem:highterm1}
Under the notation and assumptions of Theorems~\ref{Thr:optscalQuanTA} and~\ref{cor:higherorder} and Definition~\ref{def:Keynotations} then 
 \[\frac{1}{2} V(\beta) -I(\beta) =\mathcal{O}\left( \beta^{-k} \right)\]
where  in general $k=\min\left\{2+\gamma, 5/2 \right\}$ but if  $h'''(0)=0$ then $k=\min\left\{2+\gamma, 3 \right\}$.
\end{lemma}

\begin{proof}
 It has already been established that $V(\beta)=1/\beta^2$ for all distributions. Also, for a  Gaussian density, $f(\cdot)$, $I(\beta)=1/(2\beta^2)$. Since $g_\beta (\cdot)$ approaches the density of a standard Gaussian, $\phi(\cdot)$, as $\beta \rightarrow \infty$, then one expects that $I(\beta)$ would approach $1/(2\beta^2)$ too. Hence, a rigorous analysis of this convergence needs to be established. Note that
\begin{eqnarray}
	I(\beta) &=& \mathrm{Var}_\beta \left[h(X)\right] \nonumber\\
	&=& \int \left(h(x)- \mathbb{E}_\beta [h(X)] \right)^2 \frac{f^\beta(x)}{Z(\beta)} dx \nonumber\\
	&=& \int \left(h\left(\frac{y}{\sqrt{\beta (-h''(0))}}\right)- \mathbb{E}_{g_\beta} \left[h\left(\frac{y}{\sqrt{\beta (-h''(0))}}\right)\right] \right)^2 g_\beta(y) dy 		\label{eq:Ibetterm1} 
\end{eqnarray}
using the change of variable, $X=\frac{Y}{\sqrt{\beta (-h''(0))}}$.
By Taylor expansion of $h$ about the mode point, 0, up to fourth order then
\begin{eqnarray}
h\left(\frac{y}{\sqrt{\beta (-h''(0))}}\right)&=& h(0)-\frac{y^2}{2\beta }\nonumber \\
&&~~+\frac{y^3 h'''(0)}{6\left(\beta (-h''(0))\right)^{3/2}}+\frac{y^4 h''''(\xi_1(y))}{24\left(\beta (-h''(0))\right)^{2}} \label{eq:hexpan11}
\end{eqnarray}
where $\xi_1(\cdot)$ is the truncation term for the Taylor expansion such that $0<|\xi_1(y)|<\left|\frac{y}{\sqrt{\beta (-h''(0))}}\right|$ for all $y$. Using the Taylor expansion form of $h$ and the assumption of bounded fourth derivatives
\begin{eqnarray}
&&\left|\mathbb{E}_{g_\beta} \left[h\left(\frac{Y}{\sqrt{\beta (-h''(0))}}\right)-h(0)+\frac{Y^2}{2\beta }-\frac{Y^3 h'''(0)}{6\left(\beta (-h''(0))\right)^{3/2}}\right] \right|\nonumber\\
&&~~ \leq \mathbb{E}_{g_\beta} \left[\left|\frac{Y^4 h''''(\xi_1(Y))}{24\left(\beta (-h''(0))\right)^{2}}\right|\right] \leq \frac{M}{24\left(\beta (-h''(0))\right)^{2}}\mathbb{E}_{g_\beta} \left[Y^4 \right]=\mathcal{O}\left( \frac{1}{\beta^2} \right) \nonumber
\end{eqnarray}
where $\mathbb{E}_{g_\beta} \left[Y^4 \right]<\infty$ due to the assumption on the existence of moments up to the eighth moment. Thus,
\begin{eqnarray}
	\mathbb{E}_{g_\beta} \left[h\left(\frac{Y}{\sqrt{\beta (-h''(0))}}\right)\right]&=& h(0)-\frac{\mathbb{E}_{g_\beta} \left[Y^2\right] }{2\beta }+\frac{\mathbb{E}_{g_\beta} \left[Y^3 \right] h'''(0)}{6\left(\beta (-h''(0))\right)^{3/2}}\nonumber \\ &&+ \frac{\mathbb{E}_{g_\beta} \left[Y^4 h''''(\xi_1(Y))\right] }{24\left(\beta (-h''(0))\right)^{2}}, \nonumber
\end{eqnarray}
and substituting this into \eqref{eq:Ibetterm1}, along with the Taylor expansion of $h$  to the fourth order given in \eqref{eq:hexpan11}, gives
\begin{eqnarray}
I(\beta)&=& \int \Bigg( h(0)-\frac{y^2 }{2\beta }+\frac{y^3 h'''(0)}{6\left(\beta (-h''(0))\right)^{3/2}}+\frac{y^4 h''''(\xi_1(y))}{24\left(\beta (-h''(0))\right)^{2}} \nonumber\\
&&-\Bigg[  h(0)+\frac{\mathbb{E}_{g_\beta} \left[Y^2\right] h''(0)}{2\beta (-h''(0))}+\frac{\mathbb{E}_{g_\beta} \left[Y^3 \right] h'''(0)}{6\left(\beta (-h''(0))\right)^{3/2}}\nonumber \\ &&+ \frac{\mathbb{E}_{g_\beta} \left[Y^4 h''''(\xi_1(Y))\right] }{24\left(\beta (-h''(0))\right)^{2}}  \Bigg]
\Bigg)^2  g_\beta(y)dy \nonumber\\
&=& \frac{1}{4\beta^2} \int \left(  y^2 -  \mathbb{E}_{g_\beta} \left[Y^2\right] \right)^2 g_\beta(y)dy \nonumber \\
&&+\frac{2h'''(0)}{24\beta^{5/2}(-h''(0))^{3/2}} \int \left(  y^2 -  \mathbb{E}_{g_\beta} \left[Y^2\right] \right)\left(  y^3 -  \mathbb{E}_{g_\beta} \left[Y^3\right] \right)g_\beta(y)dy\nonumber\\
&&   +\mathcal{O}\left( \frac{1}{\beta^3} \right),   \nonumber
\end{eqnarray}
which is finite and well defined due to assumptions 2 and 3. Consequently, in general
\begin{eqnarray}
I(\beta)&=& \frac{1}{4\beta^2}\mathrm{Var}_{g_\beta}\left( Y^2  \right) + \mathcal{O}\left(\frac{1}{\beta^{5/2}}\right), \nonumber
\end{eqnarray}
but in the case that $h'''(0)=0$, which indeed holds in the case that $f$ is symmetric about the mode point, then
\begin{eqnarray}
I(\beta)&=& \frac{1}{2\beta^2}\mathrm{Var}_{g_\beta}\left( Y^2  \right) + \mathcal{O}\left(\frac{1}{\beta^{3}}\right) \nonumber
\end{eqnarray}
and so under the key assumption given in \eqref{eq:assumptionscale}, then 
\begin{equation}
	I(\beta) = \frac{1}{2\beta^2} + \mathcal{O}\left(\frac{1}{\beta^k}\right) \label{eq:Ibetaorder}
\end{equation}
where  in general $k=\min\left\{2+\gamma, 5/2 \right\}$ but if  $h'''(0)=0$ then $k=\min\left\{2+\gamma, 3 \right\}$, and so $\frac{1}{2} V(\beta) -I(\beta)=\mathcal{O}\left(\frac{1}{\beta^k}\right)$.

\end{proof}

\begin{lemma} \label{lem:highterm2}
Under the notation and assumptions of Theorems~\ref{Thr:optscalQuanTA} and~\ref{cor:higherorder} and Definition~\ref{def:Keynotations} then 
 \[ \frac{1}{4\beta}R(\beta)=\mathcal{O}\left( \beta^{-k} \right)\]
where  in general $k= 5/2 $ but if  $h'''(0)=0$ then $k=3$.
\end{lemma}

\begin{proof}
Recall that
	\begin{eqnarray}
	 \frac{1}{4\beta}R(\beta)&=& \frac{1}{4\beta }\mathbb{E}_{\beta}\left[ X^2h''(X)-X h'(X)\right] \nonumber\\
	&=& \frac{1}{4\beta }\mathbb{E}_{g_\beta}\Bigg[ \left(\frac{Y}{\sqrt{\beta (-h''(0))}}\right)^2h'' \left(\frac{Y}{\sqrt{\beta (-h''(0))}}\right) \nonumber\\
	&&~~~~~~~~~ -\frac{Y}{\sqrt{\beta (-h''(0))}}h' \left(\frac{Y}{\sqrt{\beta (-h''(0))}}\right)\Bigg]. \label{eq:rbetascale} 
	\end{eqnarray}
Using Taylor expansion about the mode at $0$ then
\begin{eqnarray}
h' \left(\frac{y}{\sqrt{\beta (-h''(0))}}\right) &=& h'(0) + \frac{y}{\sqrt{\beta (-h''(0))}} h''(0) + \frac{y^2}{2\beta(-h''(0))}h'''(0) \nonumber \\
&&~~~ + \frac{y^3}{6\beta^{3/2} (-h''(0))^{3/2}}h''''(\xi_2(y)),  \label{eq:Term1scale}
\end{eqnarray}
where $\xi_2(\cdot)$ is the truncation term for the Taylor expansion such that $0<|\xi_2(y)|<\left|\frac{y}{\sqrt{\beta (-h''(0))}}\right|$ for all $y$.
Also,
\begin{eqnarray}
h'' \left(\frac{y}{\sqrt{\beta (-h''(0))}}\right) &=& h''(0) + \frac{y}{\sqrt{\beta(-h''(0))}}h'''(0) \nonumber \\
&&~~~ +\frac{y^2}{2\beta^{3/2} (-h''(0))^{3/2}}h''''(\xi_3(y))  \label{eq:Term2scale}
\end{eqnarray}
where $\xi_3(\cdot)$ is the truncation term for the Taylor expansion such that $0<|\xi_3(y)|<\left|\frac{y}{\sqrt{\beta (-h''(0))}}\right|$ for all $y$. Hence, 
\begin{eqnarray}
&&\frac{y^2}{2\beta(-h''(0))}h'' \left(\frac{y}{\sqrt{\beta (-h''(0))}}\right)-\frac{y}{\sqrt{\beta (-h''(0))}}h' \left(\frac{y}{\sqrt{\beta (-h''(0))}}\right) \nonumber\\
&& = \frac{y^3}{2\left(\beta(-h''(0))\right)^{3/2}} h'''(0)+\frac{y^4}{\left(\beta(-h''(0))\right)^{2}}\left[ \frac{1}{2}h''''(\xi_3(y))- \frac{1}{6} h''''(\xi_2(y))\right]. \nonumber
\end{eqnarray}
Substituting this in to the $ \frac{1}{4\beta }R(\beta)$ term in \eqref{eq:rbetascale} 
\begin{eqnarray}
 \frac{1}{4\beta }R(\beta)&=& \frac{1}{4\beta }\mathbb{E}_{g_\beta}\Bigg[   \frac{Y^3}{2\left(\beta(-h''(0))\right)^{3/2}} h'''(0)\nonumber \\ 
&&+\frac{Y^4}{\left(\beta(-h''(0))\right)^{2}}\left[ \frac{1}{2}h''''(\xi_3(Y))- \frac{1}{6} h''''(\xi_2(Y))\right] \Bigg]\nonumber \\
&=&  \frac{h'''(0)}{8\beta^{5/2} (-h''(0))^{3/2}}\mathbb{E}_{g_\beta} \left[ Y^3 \right] \nonumber \\ 
&& + \frac{1}{4\beta^{3} (-h''(0))^{2}} \mathbb{E}_{g_\beta} \left[ Y^4  \left[\frac{1}{2}h''''(\xi_3(Y))- \frac{1}{6} h''''(\xi_2(Y))\right]\right],  \nonumber 
\end{eqnarray}
where \[ \mathbb{E}_{g_\beta} \left[ Y^4  \left[\frac{1}{2}h''''(\xi_3(Y))- \frac{1}{6} h''''(\xi_2(Y))\right]\right]   <\infty \]
due to the assumptions of boundedness of the fourth derivatives of $\log f(X)$ and the existence of moments. Hence, in general \[  \frac{1}{4\beta}R(\beta)= \mathcal{O}\left( \frac{1}{\beta^{5/2}} \right) \] but in the case that $h'''(\cdot)=0$, which is the case when $f(\cdot)$ is symmetric about the mode point 0, then  \[  \frac{1}{4\beta }R(\beta)= \mathcal{O}\left( \frac{1}{\beta^{3}} \right). \] Consequently, 
\begin{equation}
	 \frac{1}{4\beta }R(\beta) =  \mathcal{O}\left(\frac{1}{\beta^k}\right) \label{eq:orderRbet}
\end{equation}
where in general $k= 5/2$ but in the case that  $h'''(0)=0$ then $k= 3 $.
\end{proof}

\bibliographystyle{apalike}

\bibliography{biblisim}

\end{document}